\documentclass[journal,10pt,onecolumn,draftclsnofoot]{IEEEtran}
\usepackage{amsthm,amsmath,amssymb,mathtools,dsfont}
\usepackage{graphicx}
\usepackage{cite}
\usepackage{cleveref}
\usepackage{color}
\usepackage{support-caption}
\usepackage{caption}
\usepackage{subcaption}
\usepackage{algorithm}%,algorithmicx,algpseudocode}
\usepackage{algorithmic}

\newtheorem{lemma}{Lemma}

\newtheorem{claim}{Claim}
\newtheorem{theorem}{Theorem}

\newtheorem{remark}{Remark}

\usepackage{booktabs,lipsum}
\usepackage{mathtools}
\usepackage{newtxmath}
\DeclareMathOperator{\Tr}{Tr}
\DeclareMathOperator{\Id}{Id}

\ifCLASSINFOpdf

\else

\fi

\begin{document}

\title{Infinite-Horizon Linear-Quadratic-Gaussian Control with Costly Measurements}

% \author{Yunhan~Huang,~\IEEEmembership{Student Member,~IEEE}
%   % <-this % stops a space
% \thanks{Yunhan Huang and Quanyan Zhu are with the Department
% of Electrical and Computer Engineering, New York University, USA,
% NY, 11201 USA e-mail: yh.huang@nyu.edu, qz494@nyu.edu.}% <-this % stops a space
% % <-this % stops a space
% }

\author{Yunhan Huang$^{1}$ and Quanyan Zhu$^{1}$% <-this % stops a space
%\thanks{This work was not supported by any organization}% <-this % stops a space
\thanks{$^{1}$ Y. Huang and Q. Zhu are with the Department of Electrical and Computer Engineering,
        New York University, 370 Jay St., Brooklyn, NY.
        {\tt\small \{yh.huang, qz494\}@nyu.edu}}%%
}

\markboth{ArXiv Version}%
{Huang \MakeLowercase{\textit{et al.}}: A draft}



\maketitle

\begin{abstract}
In this paper, we consider an infinite horizon Linear-Quadratic-Gaussian control problem with controlled and costly measurements. A control strategy and a measurement strategy are co-designed to optimize the trade-off among control performance, actuating costs, and measurement costs. We address the co-design and co-optimization problem by establishing a dynamic programming equation with controlled lookahead. By leveraging the dynamic programming equation, we fully characterize the optimal control strategy and the measurement strategy analytically. The optimal control is linear in the state estimate that depends on the measurement strategy. We prove that the optimal measurement strategy is independent of the measured state and is periodic. And the optimal period length is determined by the cost of measurements and system parameters. We demonstrate the potential application of the co-design and co-optimization problem in an optimal self-triggered control paradigm.  Two examples are provided to show the effectiveness of the optimal measurement strategy in reducing the overhead of measurements while keeping the system performance.
\end{abstract}

% \begin{IEEEkeywords}
% IEEE, IEEEtran, journal, \LaTeX, paper, template.
% \end{IEEEkeywords}

% \fi
\IEEEpeerreviewmaketitle

\section{Introduction}
Traditional approaches to networked control systems assume the consistent availability of cost-free measurements \cite{zhivoglyadov2003networked}. Feedback control strategies are studied and designed to minimize specific cost criteria, e.g., actuating costs and the cost of deviation from the desired system state. Feedback control strategies are usually designed as a function of an estimate of the system state. The estimate is updated based on the consecutive measurements of the system outputs. The control performance relies heavily on the estimation quality, and the latter hinges on the availability and the quality of measurements. 

However, control applications in certain areas, e.g., the Internet of Things (IoT) and Battlefield Things (IoBT), may introduce a non-negligible cost of measurements. The overhead of measurements is mainly generated by 1). the price of sensing, which includes monetary expense such as power consumption and strategic cost such as stealth considerations. For example, a radar measurement can easily lead to megawatts of power usage and the exposure of the measurer to the target, and 2) the cost of communication. The cost of communication can be prohibitive for long-distance remote control tasks such as control of spacecraft and control of unmanned combat aerial vehicles. With the concern about the measurement cost raised, it is natural to ask ourselves the following questions: Can we measure less to balance the trade-off between the control performance and the cost of measurements. Hence, the high cost of measurements invokes the need for an effective and efficient measurement strategy co-designed with the control strategies to co-optimize the control performance, the cost of control, and the cost of measurement. 

Motivated by this need, we consider the co-design of the control and the measurement strategies of a linear system with additive white Gaussian noise to co-optimize a specific cost criterion over an infinite-horizon. The cost includes the traditional cost criterion in Linear-Quadratic-Gaussian (LQG) control plus the cost of measurements. The cost of an individual measurement is quantified by a time-invariant real-valued scalar $O\geq 0$. At each step, the measurement strategy provides guidelines on whether to measure based on current information at the controller's disposal. A measurement made will induce a cost quantified by $O$. If no measurement is made, there is no cost.  Control applications incorporated with Sensing-as-a-Services (SaaSs) and Communicating-as-a-Service (CaaSs) can also be framed into the binary measurement decision and the cost setting. For example, when a third party provides SaaSs with a pay-as-you-go pricing model, every time a measurement is made, a cost $O$ is paid to the third party. Here, the cost $O$ can be the price the controller pays for each sensing. The control strategy is co-designed with the measurement strategy, and controls are generated based on the measurements received.

\subsection{Related Works} The consideration of limiting the number of measurements is not new \cite{kushner1964optimum,meier1967optimal,tanaka1981suboptimal,gao2018optimal,imer2005optimal,ahmadi2018stochastic,rabi2006multiple}. Harold J. Kushner study a scalar linear-quadratic control problem when only a given number of measurements is allowed over a finite horizon \cite{kushner1964optimum}. Lewis Meier et al. generalizes the idea of \cite{kushner1964optimum} and consider the control of measurement subsystems to decide when and what to measure in a finite horizon LQG control \cite{meier1967optimal}. The idea of a limiting the number of measurements is also extended to optimal estimation problems \cite{gao2018optimal,imer2005optimal}, stochastic games \cite{ahmadi2018stochastic} and continuous-time settings \cite{rabi2006multiple}. However, instead of imposing a hard constraint on the number of measurements allowed, our work applies a soft penalty on the measurements made and study an infinite-horizon problem.

Another type of related works focuses on optimal sensor selection, where a specific combination of sensors is associated with a certain cost.  References include but is not limited \cite{athans1972determination,ko2007lqg,wu2008optimal,tzoumas2020lqg}. Readers can refer to \cite{tzoumas2020lqg} for a complete list of literature in this category. Sensor selections are either made beforehand and fixed or subject to change at each time step. The selections will decide what the controller can observe at each step. However, our work studies the decision making of when to observe instead of what to observe. Also, different from \cite{tzoumas2020lqg} where the authors study the optimal control subject to a constrained sensing budget or the optimal sensing subject to control performance constraints, we consider a co-design and co-optimization problem where the control strategy and the measurement strategy are co-designed to optimize the control performance, the control cost and the measurement cost.

The references closest to our work are  \cite{cooper1971optimal,longman1983optimal,molin2009lqg,huang2019continuous,maity2017linear, huang2020cross}.  In 70-80s, Carl Cooper et al., inspired by \cite{kushner1964optimum}, consider co-optimize the conventional cost in LQG control plus measurement costs in a finite-horizon \cite{cooper1971optimal,longman1983optimal}. The measurement cost is induced each time when a measurement is completed. \cite{molin2009lqg} solves the same problem in the networked control systems context. In \cite{cooper1971optimal,longman1983optimal,molin2009lqg}, the optimal measurement strategy can only be computed numerically based on a dynamic programming equation. Different from them, our work solves an infinite-horizon problem where both the optimal control strategy and the optimal measurement strategy are fully characterized analytically. More recently,  \cite{huang2019continuous} considers the problem of costly measurement on a continuous-time Markov Decision Process (MDP) setting. However, \cite{huang2019continuous} only establishes a dynamic programming theorem, and the characterization of optimal measurement strategy can only be carried out numerically. The consideration of costly information is also studied in finite-horizon dynamic games \cite{maity2017linear,huang2020cross}. \cite{maity2017linear} studies a two-person general sum LQG game where both players are subject to additional costs of measurements. A perfect measurement is sent to both players only when both players simultaneously choose to measure. In \cite{huang2020cross}, the authors consider a two-person zero-sum LQG game to model a cross-layer attack in an adversarial setting, where the controller chooses whether to measure, and the attacker chooses whether to jam. The actions of jamming and measuring generate costs to both players.

\subsection{Contributions}

We address a co-design and co-optimization problem of control and measurement concerning control costs and measurement costs in an infinite-horizon LQG context. The problem extends LQG control to the cases where, besides designing a control strategy and an estimator,  the controller has to decide when to measure to compensate for the overhead of measurements. The controller, consisting of a control strategy and a measurement strategy, results in a more economical control system in applications where the overhead of measurements is non-negligible. The framework also facilitates the incorporation of SaaSs and CaaSs into control systems and provides an economically efficient controller therein.

To solve the proposed co-design and co-optimization LQG problem. We first leverage an equivalent formulation with different strategy spaces in which the policies can be represented by each other and produce equal costs. We then propose a dynamic programming (DP) equation with controlled lookahead to serve as a theoretical underpinning for us to attain an optimal control strategy and an optimal measurement strategy. In \cite{cooper1971optimal,longman1983optimal,molin2009lqg}, the authors study a finite-horizon problem, and the measurement decisions need to be computed numerically beforehand. Unlike \cite{cooper1971optimal,longman1983optimal,molin2009lqg}, our work characterizes an optimal measurement strategy analytically and provides an online implementation of the derived optimal strategy. 

First, we establish the Bellman equation, which we call a dynamic programming equation with controlled lookahead. Using the Bellman equation, we show that the optimal control strategy is an open-loop strategy between two measurements. We treat the current measured state as an initial condition in each open-loop problem. The open-loop optimal control whose duration is decided by the measurement strategy is nested in a closed-loop system. We then show that the optimal measurement strategy is independent of the current measured state and can be found by solving a fixed-point equation that involves a combinatorial optimization problem. The optimal measurement strategy turns out to be periodic, and the period length is determined only by system parameters and the measurement cost. Besides, we also show how a linear-quadratic self-triggered problem \cite{gommans2014self} can be framed into the proposed dynamic programming equation with controlled lookahead. 

\textbf{Organization of the rest of the paper.} \Cref{Sec:Formulation} presents the formulation of the infinite-horizon LQG control and measurement co-design and co-optimization problem. In \Cref{Sec:TheoreicalAnalysis}, we provide the theoretical results of this paper, including the equivalent formulation, the dynamic programming equation with controlled lookahead, and the characterization of optimal strategies. \Cref{Sec: Experiements} contains two examples that help demonstrate the co-design and co-optimization problem.

\subsection{Notation}
Given any matrix $M \in \mathbb{R}^{q\times p}$, $M'$ means the transpose of the matrix $M$. When a matrix $M$ is positive semi-definite, we say $M>=0$. When a matrix is positive definite, we say $M>0$. Here, $\mathbb{R}$ is the space of real numbers and $\mathbb{N}$ is the set of natural numbers. For any given two matrices $M_1,M_2$ with the same dimension, $M_1 \geq M_2$ if $M_1 - M_2\geq 0$. For any given squared matrix $M$, $\Tr(M)$ means the trace of $M$. The identity matrix is written as $\Id$. Suppose there is a sequence of vectors $v_k$ for $k=0,2,3,\cdots K-1$, $u_{0:K-1} \coloneqq (u_0,u_1,\cdots,u_{K-1})$. Given a set $U$, $\times_k U$ means the k-ary Cartesian power of a set $U$, i.e., $\times_k U \coloneqq\underbrace{ U \times U \times \cdots \times U }_{k}$. 

\section{Formulation}\label{Sec:Formulation}
In the discrete-time Gauss-Markov setting, we consider the following linear dynamics of the state $x_t$:
\begin{equation}\label{Eq:SystemDynamics}
\begin{aligned}
    x_{t+1} &= A x_t + B u_t + Cw_t,\\
    y_t &= i_n x_n,
\end{aligned}
\end{equation}
where $x_{t} \in \mathcal{X} = \mathbb{R}^q$ is the state at time $t$, and $u_t \in \mathcal{U} = \mathbb{R}^p$, with dimension $p$ lower than or equal to $q$, is the control at time $t$. Here, $w_t$ is the Gaussian noise with zero mean and $\mathbb{E}[w_t w_\tau']= \Sigma_s \delta_{t-s}$, where $\delta_t$ is the Kronecker delta. We have the standard assumption that $C'\Sigma_S C$ is positive definite. That is to say system noises are linearly independent. The matrices $A$, $B$ and $C$ are real-valued with proper dimension. The measurement decision at time $t$ is denoted by $i_t\in \{0,1\}$, which be called the measurement indicator. A meaningful measurement $y_n = x_n$ is made only when $i_n$ is one. The initial condition $x_0$ is assumed to be known by the controller. 

The cost functional associated with \Cref{Eq:SystemDynamics} is given as 
\begin{equation}\label{Eq:CostFunctional}
    F(\pi;x) = \mathbb{E} \left[ \sum_{t=0}^\infty \beta^t( x_t'Q x_t + u_t'Ru_t + i_t O) \middle\vert x_0 = x\right],
\end{equation}
where we assume that $Q \equiv Q'$ is positive semi-definite, $R\equiv R'$ is positive definite and both $Q$ and $R$ are with proper dimension. Here, $O\in \mathbb{R}^+$ is the nonnegative cost of measurement, $\beta<1$ is the discount factor, and $\pi$ is a notation for the strategy that will be defined shortly. We introduce the notation to denote the history of variables
\begin{equation}\label{Eq:InformationStructure}
I_t = \{i_0,\dots,i_t\},\;\;\; U_t = \{u_0,\dots,u_t\},\;\;\;  Y_t =\{y_0,\dots,y_t\}.
\end{equation}
We define $\mathcal{F}_t =\{ I_{t-1}, U_{t-1}, Y_{t-1}, x_0\}$ and $\bar{\mathcal{F}}_t = \{\mathcal{F}_t, i_t, y_t\}$ as the information available to the controller at time $t$ before and after a measurement decision is made. The measurement decision is made based on $\mathcal{F}_{t}$ and the control is decided based on $\bar{\mathcal{F}_t}$. Hence, our objective is to find the stationary strategy $\pi = (\mu,\nu)$ that generates a sequence of measurement decisions $\{i_t = \mu(\mathcal{F}_t),t=0,1,\cdots\}$ and a sequence of controls $\{u_t = \nu(\bar{\mathcal{F}}_t),t=0,1,\cdots\}$ to minimize \Cref{Eq:CostFunctional}. We define $\Pi$ as the space of all such strategies. In this formulation, i.e., the formulation defined by \Cref{Eq:SystemDynamics,Eq:CostFunctional}, the controller decides whether to measure at every time step. In next section, we propose an equivalent formulation that facilitates the process of finding an optimal measurement strategy and a control strategy.

%The exact sequence of events at any time $t$ can be summarized as: 1) controller generates $i_t$; 2) controller observes $y_t$; 3) controller generates $u_t$; 4) system generates $x_{n+1}$.

\section{Theoretical Analysis}\label{Sec:TheoreicalAnalysis}

In this section, we find the optimal strategies $\pi$ by following two steps. The first step is to formulate an equivalent representation of the original problem defined by \Cref{Eq:SystemDynamics,Eq:CostFunctional}. In the second step, we propose a dynamic programming equation with controlled lookahead based on the representation problem, which serves as a theoretical underpinning to characterize the optimal strategies. 

\subsection{An Equivalent Representation}
The representation has the following cost functional associated with \Cref{Eq:SystemDynamics}:
\begin{equation}\label{Eq:CostFunctionalRep}
\tilde{F}(\tilde{\pi};x) = \mathbb{E}\left[ \sum_{t=0}^\infty \left(x_t'Q x_t + u_t'R u_t \right) + \sum_{k=1}^\infty \beta^{\bar{T}_k} O \middle \vert x_0 =x \right],
\end{equation}
which is associated with the stationary strategy $\tilde{\pi} \in \tilde{\Pi}: \mathbb{\mathcal{X}}\rightarrow \mathbb{N}\times \mathcal{U} \times \mathcal{U} \times \cdots$. Here, $t$ is the index of time steps and $k$ is a counter of the number of measurements.  Basically, at time $t$ when a measurement is made, a strategy $\tilde{\pi}$ prescribes a waiting time for next measurement $T$ and a sequence of controls between two observation epochs $(u_t,u_{t+1},\dots,u_{t+T-1})$ based on current observation $x_t$. That is $(T,u_{t},\cdots,u_{t+T-1}) = \tilde{\pi}(x_t)$. To facilitate discussion, $T_k$ is denoted as the waiting time before the $k$th measurement. In \Cref{Eq:CostFunctionalRep}, $\bar{T}_k$ is the time instance of the $k$th measurement defined as $\bar{T}_k = \sum_{i \leq k} T_i$ and $\bar{T}_0=0$. That is at $t=\bar{T}_k$, the $k$th measurement is made. Since $x_0$ is known to the controller, the first measurements happens at time $\bar{T}_1 = T_1$. To facilitate the readers, corresponds between $T_k$, $\bar{T}_k$ and the measurement indicators $I_{t}$ defined in \Cref{Eq:InformationStructure}, are illustrated in \Cref{fig:MeasurementTimeIllustration}. Next, we show, using \Cref{Lemma:EquivalentRepresentation}, that by finding an optimal strategy $\tilde{\pi}^*\in\tilde{\Pi}$ of the problem defined by \Cref{Eq:CostFunctionalRep}, we can find an optimal strategy $\pi^*\in \Pi$ of the problem defined by \Cref{Eq:SystemDynamics}.

\begin{figure}[h]
    \centering
    \includegraphics[width=0.9\columnwidth]{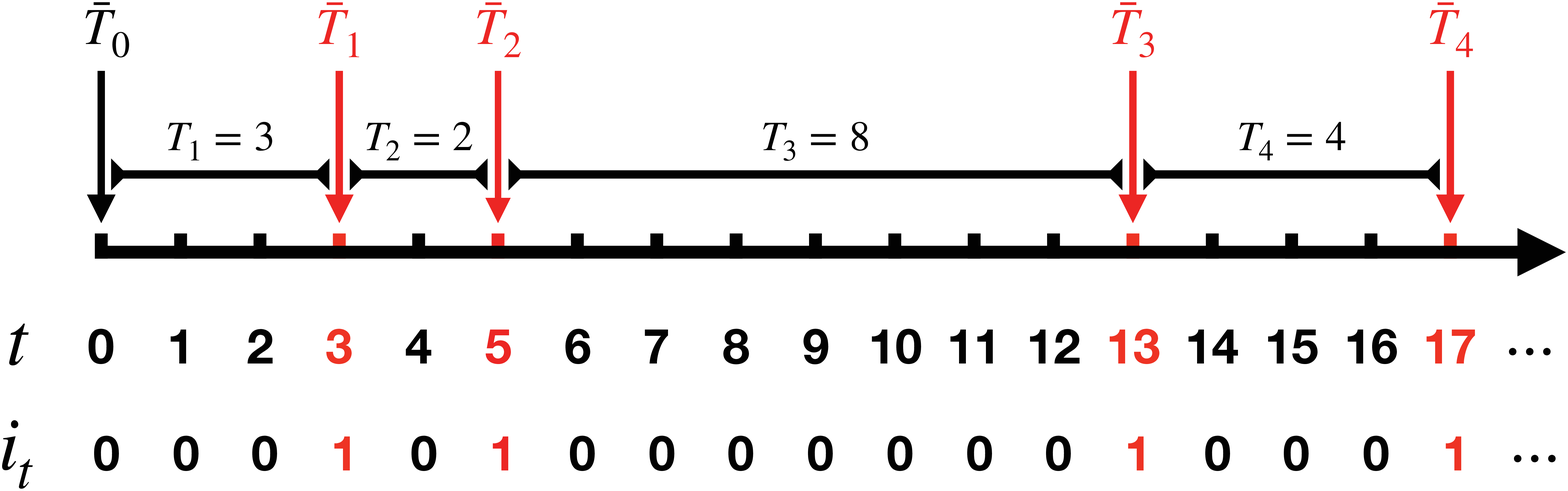}
    \caption{An illustration of the relations between $I_t$, the measurement indicators, and $\bar{T}_{k-1},T_{k}$, the time instance of the $k-1$th measurement and the waiting time for the $k$th measurement.}
    \label{fig:MeasurementTimeIllustration}
\end{figure}

\begin{lemma}\label{Lemma:EquivalentRepresentation}
The infinite-horizon LQG control problem with costly measurements defined by \Cref{Eq:CostFunctional} associated with strategy $\pi\in\Pi$ can be equivalently represented by the optimal control problem defined by $\Cref{Eq:CostFunctionalRep}$ associated with strategy $\tilde{\pi}\in \tilde{\Pi}$. That is every strategy $\pi\in \Pi$ can be represented by a strategy $\tilde{\pi}$ in $\tilde{\Pi}$ (See Section 5.6 of \cite{bacsar1998dynamic} for representations of strategies) and they both produce the same cost, and vice versa;
\end{lemma}
\begin{proof}
See \Cref{Proof:EquivalentRepresentation}.
\end{proof}

\begin{remark}
An strategy $\pi$ corresponding to \Cref{Eq:SystemDynamics} and a strategy $\tilde{\pi}$ corresponding to \Cref{Eq:CostFunctionalRep} can be interpreted as different system implementations. For $\pi$ in \Cref{Eq:SystemDynamics}, at the beginning of time $t$, 1). the controller decides whether to measure according to $\mu(\bar{\mathcal{F}}_t)$. 2). If the decision is to measure, the controller sends a request to the measurement system and receives $y_n = x_n$. Otherwise, no request is sent and no information is received by the controller. 3). Then the control command is then computed based on $\mu(\bar{\mathcal{F}}_t)$ and sent to the actuators. 4). The system then generates $x_{t+1}$. For $\tilde{\pi}$ in \Cref{Eq:CostFunctionalRep}, at $t = \bar{T}_k$, 1) the controller receives its $k$th measurement $y_{\bar{T}_k}= x_{\bar{T}_k}$ from the measurement system.  2) The controller computes the waiting time for next measurement $T_{k+1}$ and a sequence of control commands $(u_{\bar{T}_k},\cdots,u_{\bar{T}_k + T_{k+1}-1})$. 3) The waiting time $T_{k+1}$ is sent to the measurement system indicating the next time to measure and the sequence of control commands is sent to the actuator, either in one packet or in $T_k$ packets over time. 4) The actuators apply these commands and the system updates $x_{\bar{T}_{k}+1},\cdots, x_{\bar{T}_k + T_{k+1}+1}$.
\end{remark}

\subsection{Dynamic Programming Equation with Controlled Lookahead}
With \Cref{Lemma:EquivalentRepresentation}, we thus can focus on analyzing the representation problem defined by \Cref{Eq:SystemDynamics,Eq:CostFunctionalRep} and characterizing the optimal strategy $\tilde{\pi}^*$ therein. To begin with, we are interested in minimizing the cost functional over the entire space of policies taking the form $\tilde{\pi}: \mathbb{\mathcal{X}}\rightarrow \mathbb{N}\times \mathcal{U} \times \mathcal{U} \times \cdots$. The values of the infimum is defined as 
\begin{equation}\label{Eq:ValueFunction}
V(x) \coloneqq \inf_{\tilde{\pi}\in \tilde{\Pi}} \tilde{F}(\tilde{\pi};x) = \inf_{\pi\in\Pi} F(\pi;x).
\end{equation}

The following theorem shows the dynamic programming equation regarding the value functions defined in \Cref{Eq:ValueFunction}, which we call the dynamic programming equation with controlled lookahead. The proof of the theorem is based on the idea of consolidating the induced costs and the generated controls between measurement epochs and formulating an MDP problem with extended state and action spaces.  

\begin{theorem}\label{Theorem:DynamicProgramming}
The value function $V(x)$ defined by \cref{Eq:ValueFunction} satisfies the following dynamic programming equation
\begin{equation}\label{Eq:DynamicProgramming}
V(x) = \inf_{T\in\mathbb{N}} \inf_{u_{0:T-1} \in \times_T \mathcal{U}} \mathbb{E} \left[ \sum_{t=0}^{T-1} \beta^t (x_t' Q x_t + u_t' R u_t) + \beta^T V(x_T) + \beta^T O \middle \vert x_0 =x \right].
\end{equation}
If there exists a strategy $\tilde{\pi}^*(x) = (T^*,u_0^*,\cdots,u_{T-1}^*)$ such that 
$$
V(x) = \mathbb{E} \left[ \sum_{t=0}^{T^*-1} \beta^t (x_t' Q x_t + {u_t^*}' R u^*_t) + \beta^{T^*} V(x_{T^*}) + \beta^{T^*} O \middle \vert x_0 =x \right],
$$
for all $x\in\mathcal{X}$, then $\tilde{\pi}^*$ is the optimal strategy.
\end{theorem}
\begin{proof}
See \Cref{Proof:DynamicProgramming}.
\end{proof}

\begin{remark}
The dynamic programming involves the consolidated stage cost $\sum_{t=0}^{T-1} \beta^t (x_t' Q x_t +  u_t' R u_t)$, the cost-to-go after $T$-steps lookahead, and the cost of next measurement. Hence, the dynamic programming equation has $T$-steps lookahead and the number of steps $T$ is controlled and optimized according to the trade-off between the control performance degradation and the measurement cost. We thus refer to the dynamic programming equation in \Cref{Eq:DynamicProgramming} as the dynamic programming equation with controlled lookahead, which differs from the traditional lookahead dynamic programming equations \cite{bertsekas1995dynamic} in two ways. The first is that the number of lookahead steps is controlled. The second is that the control strategy is dependent solely on $x$ (no closed-loop state updates) and will be applied in the next $T$ steps.
\end{remark}

\subsection{The Optimal Measurement and Control Strategies}
From \Cref{Theorem:DynamicProgramming}, we know that the characterization of the optimal policy relies on solving the dynamic programming equation given in \Cref{Eq:DynamicProgramming} which is basically a fixed-point equation. The uniqueness of the value function is guaranteed by the Banach fixed-point theorem \cite{kreyszig1978introductory} using the fact that the operator defined by the right-hand side of \Cref{Eq:DynamicProgramming} is a contraction mapping. To calculate the right hand-side of \Cref{Eq:DynamicProgramming} for a given $V(x)$, one can first fix $T$ and treat the inner minimization problem in \Cref{Eq:DynamicProgramming} as an open-loop optimal control problem starting at $x_0=x$ with terminal cost $\beta^T V(x^T)$, which gives the following lemma.

\begin{lemma}\label{Lemma:InnerOptimalControl}
Suppose that $V(x) = x' P x  + r$, where $P$ is a real-valued matrix with proper dimension and $r$ is a real-valued scalar. Given any $T$, the inner optimization problem in \Cref{Eq:DynamicProgramming}
$$
\inf_{u_0,\cdots,u_{T-1}} \mathbb{E} \left[ \sum_{t=0}^{T-1} \beta^t (x_t' Q x_t + u_t' R u_t) + \beta^T x_T'Px_T + \beta^T r + \beta^T O \middle \vert x_0 =x \right]
$$
has the minimum (the optimal cost) 
$$
f_0^*(x) = x' L_{T} x + \sum_{t=0}^{T-1} \beta^t \Tr\left( P_{t}(\bar{\mathcal{F}}_t)  \varphi_{t}\right) + \sum_{t=1}^{T} \beta^t \Tr \left( \Sigma_S C L_{T-t}C\right) + \beta^T (r+O),
$$
where $L_t$ is generated by the Riccati equation
\begin{equation}\label{Eq:RiccatiEquation}
L_{t+1} = Q + \beta A'L_t A - A' L_t B \beta (R+\beta B'L_t B)^{-1}\beta B'L_t A,\ \ \ \textrm{for }t=0,\cdots,T,\ \textrm{with }L_0 =P,
\end{equation}
and $\varphi_t$ is generated according to
\begin{equation}\label{Eq:EstimationErrorCoefficient}
\varphi_t = A' L_{T-t-1}B \beta (R+\beta B' L_{T-t-1} B)^{-1}\beta B' L_{T-t-1}A,\ \ \ \textrm{for }t=0.\cdots,T-1.
\end{equation}
The corresponding minimizer (the optimal controls) is
$$
\begin{aligned}
u^*_t &= -(R + \beta B' L_{T-t-1}B)^{-1} \beta B' L_{T-t-1} A \hat{x}_t\\
\end{aligned}
$$
Here, $P_{t}(\bar{\mathcal{F}}_t) = \mathbb{E}\left[(x_t - \hat{x}_t)'(x_t - \hat{x}_t)\middle \vert \bar{\mathcal{F}}_t\right]$ the covariance of estimation error when no measurement is made from $t=1$ to $t= T-1$. And $\hat{x}_t \coloneqq \mathbb{E}\left[ x_t\middle \vert \bar{\mathcal{F}}_t \right]$ is the estimate of $x_t$. The the estimate and the covariance of estimation error evolves according to
\begin{equation}\label{Eq:EstimateCovariancePropa}
\begin{aligned}
\hat{x}_{t+1} &= A\hat{x}_t + Bu^*_t,\ \ \ \textrm{with } \hat{x}_0 =x_0,\\
P_{t+1}(\bar{\mathcal{F}}_{t+1}) &=  A' P_t(\bar{\mathcal{F}}_{t})A + C'\Sigma_S C,\ \ \ \textrm{with } P_0(\bar{\mathcal{F}}_0)=0,\ \ \ \textrm{for }t=0,\cdots,T-1.
\end{aligned}
\end{equation}
\end{lemma}

\begin{proof}
See \Cref{Proof:InnerOptimalControl}.
\end{proof}

From \Cref{Lemma:InnerOptimalControl}, we know that if the value function takes the form of $x'P x' +r$, the dynamic programming equation with controlled lookahead, a.k.a. \Cref{Eq:DynamicProgramming}, can be written as
\begin{equation}\label{Eq:DPInnerSolved}
x'Px + r  = \inf_{T\in\mathbb{N}} \left\{ x' L_T x +\sum_{t=0}^{T-1} \beta^t \Tr\left( P_{t}(\bar{\mathcal{F}}_t)  \varphi_{t} \right)+ \sum_{t=1}^{T} \beta^t \Tr\left( \Sigma_S C' L_{T-t}C \right) + \beta^T (r+O) \right\}.
\end{equation}
To fully characterize the value functions, one needs to find a real-valued matrix $P$ such that $P=L_{T^*}$, where $T^*$ is the optimal waiting time for next measurement. In the following theorem, we show that the value function $V(x)$ can be solved analytically and the optimal measurement policy is independent of $x$.

\begin{lemma}\label{Lemma:ValueFunctionCharacterization}
Write $Q=J'J$. Let $(A,B)$ be controllable and $(A,J)$ be observable.  The value function defined in \Cref{Eq:ValueFunction} is $V(x) = x'Px+ r$, where $P$ is a unique solution of the following algebraic Riccati equation
\begin{equation}\label{Eq:AlgebraicRiccatiEquation}
    P = Q + \beta A'P A - A' P B \beta (R+\beta B'P B)^{-1}\beta B'P A,
\end{equation}
and $P$ is positive definite. Here, $r$ is the unique solution of the following fixed-point equation
\begin{equation}\label{Eq:FixedPointEquationForr}
r  =  \inf_{T\in\mathbb{N}} \left\{\sum_{t=0}^{T-1} \beta^t \Tr\left( P_{t}(\bar{\mathcal{F}}_t)  \varphi \right) + \sum_{t=1}^{T} \beta^t \Tr\left( \Sigma_S C' PC \right)+ \beta^T (r+O) \right\}.
\end{equation}
\end{lemma}

\begin{proof}
See \Cref{Proof:ValueFunctionCharacterization}.
\end{proof}

\Cref{Lemma:ValueFunctionCharacterization} shows that the value function is indeed quadratic in $x$ and $P$ is a positive definite matrix that satisfies the algebraic Riccati equation \Cref{Eq:AlgebraicRiccatiEquation}. The quadratic term of $x'Px$ in the value function $V(x)$ is the same as regular (no measurement cost) discounted infinite-horizon linear quadratic optimal control problem. And the optimal waiting time for next observation $T^*$, which is the minimizer of \Cref{Eq:FixedPointEquationForr}, is independent of $x$. To obtain the optimal policy, it remains to characterize $r$.

\begin{theorem}\label{Theorem:CharacterizationRandOptimalObservation}
Suppose that conditions in \Cref{Lemma:ValueFunctionCharacterization} hold, i.e., $(A,B)$ be controllable and $(A,J)$ be observable. Let $\varphi = A' PB \beta (R+\beta B'P B)^{-1}\beta B' PA$. The optimal measurement policy and the value of $r$ can be characterized as
\begin{enumerate}
    \item If the cost of measurement $O < \Tr\left( C\Sigma_S C \varphi\right)$, the optimal measurement policy is to observe every time, i.e., $T^* =1$. The solution of \Cref{Eq:FixedPointEquationForr} is 
    $$
    r=\frac{\beta}{1-\beta} \Tr\left( \Sigma_S C'PC \right)+ \frac{\beta}{1-\beta} O.
    $$ 
    The value function is 
    $$V(x) = x'Px + \frac{\beta}{1-\beta} \Tr\left( \Sigma_S C'PC\right) + \frac{\beta}{1-\beta} O. $$
    
    \item Given the cost of measurement $O$, the optimal policy is to wait $T^*$ steps for next measurement and $T^*$ can be determined by
    \begin{equation}\label{Eq:DetermineOptimalT}
    \sum_{t=0}^{T^*-2} \frac{1-\beta^{t+1}}{1-\beta} \Tr\left( (A')^tC' \Sigma_S CA^t \varphi\right) \leq O < \sum_{t=0}^{T^*-1} \frac{1-\beta^{t+1}}{1-\beta} \Tr\left( (A')^tC' \Sigma_S CA^t \varphi\right).
    \end{equation}
    The solution of \Cref{Eq:FixedPointEquationForr} is 
    $$
    r = \frac{\sum_{t=0}^{T^*-1} \beta^t \Tr\left( P_{t}(\bar{\mathcal{F}}_t)\varphi\right)}{1-\beta^{T^*}} + \frac{\beta}{1-\beta} \Tr\left( \Sigma_S C' PC\right) + \frac{\beta^{T^*}}{1-\beta^{T^*}} O,
    $$
    where the $P_t(\bar{\mathcal{F}}_t)$ is propagated according to \Cref{Eq:EstimateCovariancePropa}. The value function is 
    \begin{equation}\label{Eq:CharacterizedValueFunction}
    V(x) = x'P x + \frac{\sum_{t=0}^{T^*-1} \beta^t \Tr\left( P_{t}(\bar{\mathcal{F}}_t)\varphi\right)}{1-\beta^{T^*}} + \frac{\beta}{1-\beta} \Tr\left( \Sigma_S C' PC\right) + \frac{\beta^{T^*}}{1-\beta^{T^*}} O.
    \end{equation}
    \item If $A$ is table, there exists a unique solution $W_\infty$ of the Lyapunov function
    \begin{equation}\label{Eq:LypFun}
    W_\infty - A' W_\infty A = C'\Sigma_S C.
    \end{equation}
    If, in addition, $ O\geq \frac{\Tr\left( W_\infty \varphi\right)}{1-\beta} -  \sum_{t=0}^\infty \beta^t \Tr\left( P_t(\bar{\mathcal{F}}_t)\varphi\right)$, the optimal measurement policy is not to measure at all, i.e., $T^* = \infty$. The value function then will be
    $$
    V(x) = x'Px + \sum_{t=0}^\infty \beta^t \Tr\left( P_t(\bar{\mathcal{F}}_t) \varphi\right) + \frac{\beta}{1-\beta} \Tr\left( \Sigma_S C'PC\right).
    $$
    Otherwise, $T^*$ is finite and can be determined by 2).
\end{enumerate}
\end{theorem}
\begin{proof}
See \Cref{Proof:CharacterizationRandOptimalObservation}
\end{proof}

\begin{remark}
From \Cref{Lemma:ValueFunctionCharacterization}, we know that the optimal policy is independent of the current observed state. Hence, the optimal measurement policy is to measure periodically. The optimal measurement policy is then determined by the optimal inter-measurement time $T^*$, which can be computed according to \Cref{Theorem:CharacterizationRandOptimalObservation}. Thus, the optimal policy can be written as 
\begin{equation}\label{Eq:CharacterizationofOptimalPolicy}
\tilde{\pi}(x) = (T^*, -K x, -K(A-BK)x, \cdots, -K(A-BK)^{T^*-1}x ),
\end{equation}
where $K = (R + \beta B' P B)^{-1} \beta B' P A$. Different from \cite{huang2019continuous} in which continuous-time Markov decision process with costly measurement is studied and the optimal measurement policy depends on the current observed state, the optimal policy is independent of the current observed state in the infinite-horizon LQG setting. This is due to the linearity of the system and the Gaussian noise that can be fully characterized by its mean and covariance.
\end{remark}

\begin{remark}
From \Cref{Eq:DetermineOptimalT} and \Cref{Eq:CharacterizationofOptimalPolicy}, we can characterize the optimal strategy $\pi^* = (\mu^*,\nu^*)\in\Pi$ for the original problem defined by \Cref{Eq:CostFunctional}. Given the measurement history $I_{t-1}$, let $m_\tau$ be the number steps since the last measurement times instance and $\bar{P}_\tau$ be the surrogate covariance that are updated according to
\begin{equation}\label{Eq:SurrogateVariables}
\begin{aligned}
m_{\tau} &= \begin{cases}
0,\ \ \ &\textrm{if } i_\tau =1,\\
m_{\tau-1} + 1,\ \ \ &\textrm{if } i_\tau =0,
\end{cases}\\
\bar{P}_{\tau} &= \begin{cases}
0,\ \ \ &\textrm{if }i_\tau =1\\
\bar{P}_{\tau-1} + \frac{1-\beta^{m_{\tau-1}+1}}{1-\beta} (A')^{m_{\tau-1}} C'\Sigma_S CA^{m_{\tau-1}},\ \ \ &\textrm{if }i_\tau =0,
\end{cases}
\end{aligned}
\end{equation}
for $\tau = 1,2,\cdots,t-1$ with $m_0 = 0$  and $\bar{P}_0 =0$. Note that $I_{t-1} \subset \mathcal{F}_t$. The optimal measurement can then be written as
$$
i^*_t =\mu^*(\mathcal{F}_t) = \begin{cases}
1,\ \ \ &\textrm{if } \Tr\left( \left[\bar{P}_{t-1} +\frac{1-\beta^{m_{t-1}+1}}{1-\beta} (A')^{m_{t-1}} C'\Sigma_S CA^{m_{t-1}} \right] \varphi\right) > O,\\
0,\ \ \ &\textrm{Otherwise}.
\end{cases}
$$
Given the measurement history $I_{t}$ and the control history $U_{t-1}$, define the estimate $\bar{x}_t$ as
$$
\bar{x}_\tau = \begin{cases}
x_\tau,\ \ \ &\textrm{if }i_\tau =1,\\
A \bar{x}_{\tau-1} + B u_{\tau-1},\ \ \ &\textrm{if }i_\tau =0,
\end{cases}
$$
for $\tau =1,2,\cdots,t$ with $\bar{x}_0 =x_0$.
Note that $I_t\cup U_{t-1} \subset \bar{\mathcal{F}}_t$. The optimal control strategy can then be written as 
$$
u^*_t = \nu^*(\bar{\mathcal{F}}_t) = -K\bar{x}_\tau.
$$
Note that in \Cref{Eq:SurrogateVariables},  the term $(A')^{m_{\tau-1}} C'\Sigma_S CA^{m_{\tau-1}}(1-\beta^{m_{\tau-1}+1})/(1-\beta) $ can be updated recursively. Hence, $m_t$, $\bar{P}_t$ and $\hat{x}$ can be updated recursively, so there is no need to keep the history of them. This provides an online implementation of the results in \Cref{Lemma:InnerOptimalControl} and \Cref{Theorem:CharacterizationRandOptimalObservation}.
\end{remark}

\begin{remark}
When there is not cost of measurement, i.e., $O=0$, the problem reduces to the classic discounted infinite-horizon LQG problem \cite{bertsekas1995dynamic}. \Cref{Theorem:CharacterizationRandOptimalObservation} tells that it is optimal to measure every time, i.e., $T^* =1$. The value function is $V(x) = x'Px + \frac{\beta}{1-\beta} \Tr\left( \Sigma_SC'PC\right)$, which is the same as the value function of the classic discounted infinite-horizon LQG problem \cite{bertsekas1995dynamic,gommans2014self}. The optimal measurement policy is to not measure at all only when $A$ is stable and $ O\geq \frac{\Tr\left( W_\infty \varphi\right)}{1-\beta} -  \sum_{t=0}^\infty \beta^t \Tr\left( P_t(\bar{\mathcal{F}}_t)\varphi\right)$. Here, $P_{t}(\bar{\mathcal{F}}_t)$ is propagated according to \Cref{Eq:EstimateCovariancePropa}, who can also be expressed by the closed-form expression
$$
P_t(\bar{\mathcal{F}}_t) = \sum_{\tau=0}^{t-1} (A')^\tau C' \Sigma_S C A^\tau.
$$
\end{remark}

\begin{remark}
The framework of LQG control with costly measurements can naturally be applied to optimal self-triggered control paradigm \cite{gommans2014self,akashi2018self} considering their similar purposes of reducing the cost of sensing and the cost of communication. In an optimal self-triggered control paradigm, a fixed control between two measurements is considered in most cases. In \cite{gommans2014self}, the authors also discuss the case when multiple control commands are allowed in one packet, i.e., instead of applying a fixed control command, a sequence of time-varying control commands between two measurement instances. If multiple control commands are allowed in one packet, the optimal strategy in \Cref{Eq:CharacterizationofOptimalPolicy} can be used to implement an optimal self-triggered control paradigm. If only a single control command is allowed in one packet, we need to look into the policies
$\tilde{\pi}_f\in \tilde{\Pi}_f \subset \tilde{\Pi}$, where
$$
\tilde{\Pi}_f \coloneqq \left\{ \tilde{\pi} \in \tilde{\Pi}\  \middle\vert\ (T,u_{0:T-1}) = \tilde{\pi}(x)\ \textrm{satisfying }  u_0=u_1=\cdots=u_{T-1}\ \textrm{for all }x\in\mathcal{X} \right\}.
$$
Define the value function of the fixed control problem as $V_f(x)\coloneqq \inf_{\tilde{\pi}_f} \tilde{F}(\tilde{\pi}_f;x)$. Following the proof of \Cref{Theorem:DynamicProgramming}, we have
$$
V_f(x) = \inf_{T\in\mathbb{N}} \inf_{u\in \mathcal{U}} \mathbb{E}\left[ \sum_{t=0}^{T-1} \beta^t (x_t'Q x_t + u'Ru) + \beta^T V_f(x_T) + \beta^T O\middle \vert x_0 =x \right].
$$
Then, to find the optimal strategy, we need to find a strategy $\tilde{\pi}_f^*(x) = (T^*,u^*,\cdots,u^*)$ such that 
$$
V_f(x) = \mathbb{E} \left[ \sum_{t=0}^{T^*-1} \beta^t (x_t' Q x_t + {u^*}' R u^*) + \beta^{T^*} V_f(x_{T^*}) + \beta^{T^*} O \middle \vert x_0 =x \right].
$$
Here, we leave the characterization of the value function $V_f$ and the optimal strategy $\tilde{\pi}^*_f$ for future works. We can see that once $\tilde{\pi}_f^*$ is characterized, it can be implemented in the self-triggered control paradigm that only allows one control command in one control packet. And $\tilde{\pi}_{f}^*$ will optimize the trade-off between the control performance and the communication/sensing overhead. 
\end{remark}

In this section, we fully characterize the optimal measurement strategies and the optimal control strategies for both the original problem and its representation. Different implementation schemes are discussed. We also shed some light on the potential application of the LQG control with costly measurements framework in optimal self-triggered control. In the next section, we show how the optimal measurement strategy is determined by the cost of measurements and the dynamic behavior of certain systems under the optimal control and measurement strategies.

\section{Experiments} \label{Sec: Experiements}
In this section, we demonstrate the effectiveness of the optimal measurement strategy in reducing the overhead of measurements while keeping the system performance.  We explore two examples: one is with a Schur usntable system matrix $A_1$ and one is with a Schur stable matrix $A_2$.

The two systems, called \textbf{sys1} and \textbf{sys2}, are with system matrices 
$$
\begin{aligned}
A_1 = \begin{bmatrix}
-0.61 & 0.53 & 1.3\\
-1.15 & -0.03 & -0.96\\
-0.78 & 0.24 & -0.02
\end{bmatrix},\ \ \ 
\end{aligned}
A_2 = \begin{bmatrix}
-0.61 & 0.53 & 0.3\\
-0.95 & -0.03 & -0.56\\
-0.78 & 0.24 & -0.02
\end{bmatrix}.
$$
Other system parameters of the two systems are set to be the same. Namely,
$$
B = \begin{bmatrix}
0.12 & -0.55\\
0.86 & 0.08\\
1.16 & -0.60\\
\end{bmatrix},\ \ \ C = \begin{bmatrix}
1 & 0 & 0\\
0 & 1 & 0\\
0 & 0 & 1
\end{bmatrix},\ \ \ \sigma = 0.08\cdot\begin{bmatrix}
1 & 0 & 0\\
0 & 1 & 0\\
0 & 0 & 1
\end{bmatrix}.
$$
Suppose the initial condition is given as $x_0 = [20\ -15\ 10]'$.
The magnitudes of the three eigenvalues of $A_1$ are $(1.3561,1.3561,0.0791)$. Hence, $A_1$ is Schur unstable.  The magnitudes of the three eigenvalues of $A_2$ are $(0.9755,0.9755,0.0669)$.  Hence $A_2$ is Schur stable. It is easy to see that both \textbf{sys1} and \textbf{sys2} are controllable.

The cost parameters are given as $Q=0.1\cdot \Id$, $R=0.2\cdot \Id$ and $\beta =0.95$. Here, $\Id$ represents the identity matrix with a proper dimension.  The cost of measurement $O$ is subject to change.

\begin{figure}
     \centering
     \begin{subfigure}[b]{0.32\textwidth}
         \centering
         \includegraphics[width=\textwidth]{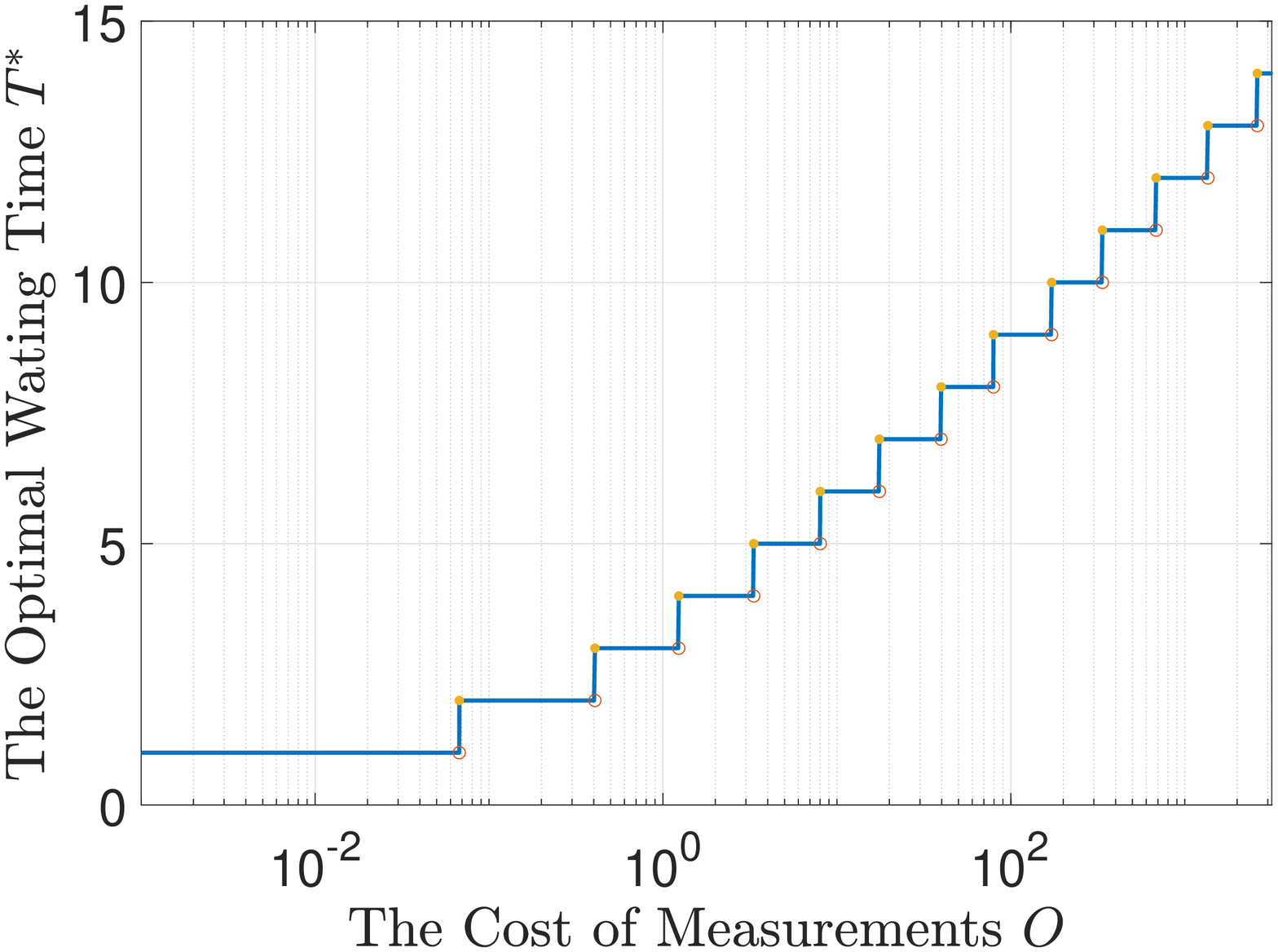}
         \caption{{\footnotesize An illustration of how the optimal waiting time $T^*$ is affected by the cost of measurements $O$.}}
         \label{fig:OVersusTstar}
     \end{subfigure}
     \hfill
     \begin{subfigure}[b]{0.32\textwidth}
         \centering
         \includegraphics[width=\textwidth]{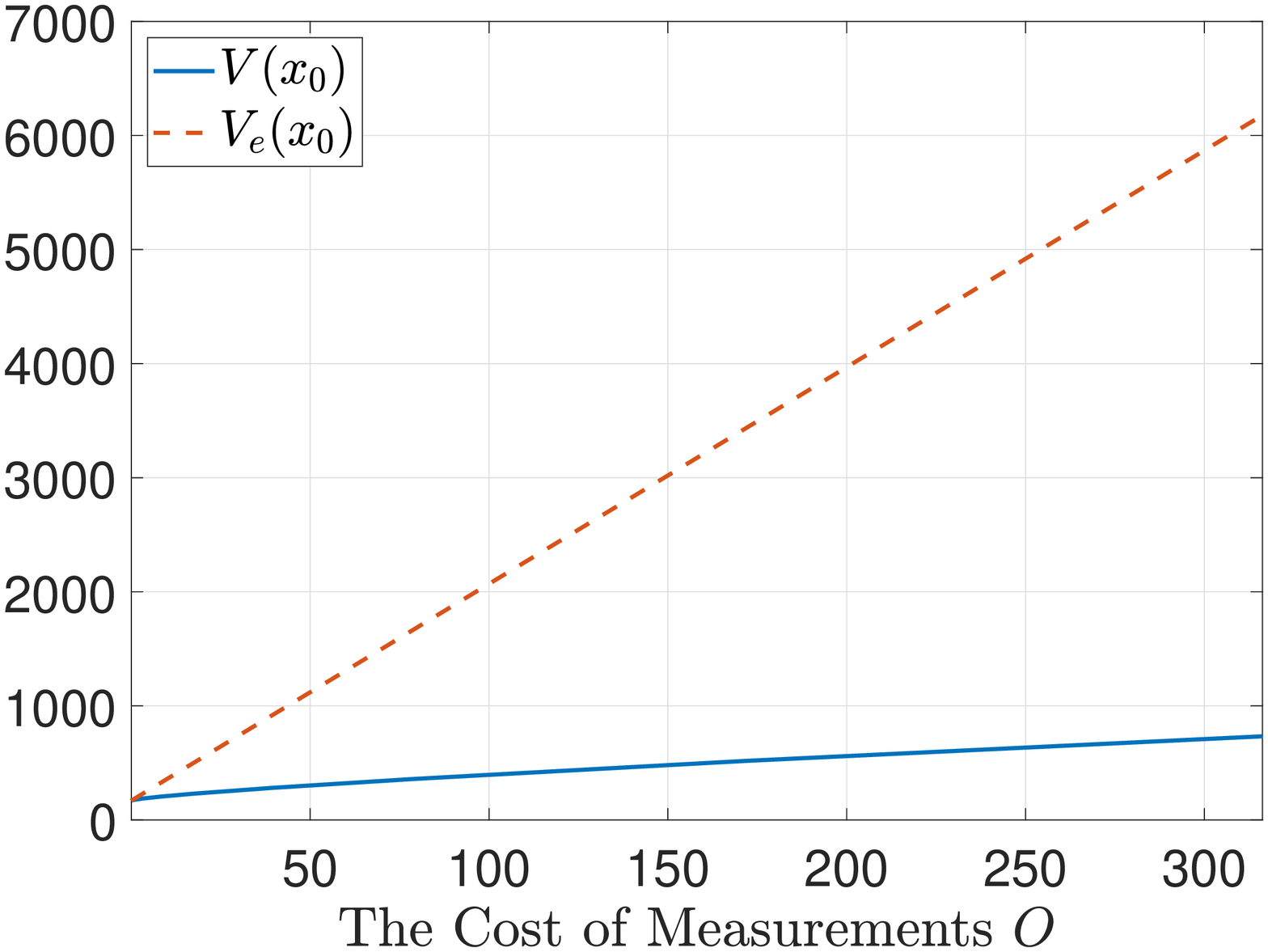}
         \caption{{\footnotesize An illustration of how the value $V(x_0)$ of the problem increases as the cost of measurements $O$ increases. }}
         \label{fig:ComparedTwoStrategies}
     \end{subfigure}
     \hfill
     \begin{subfigure}[b]{0.32\textwidth}
         \centering
         \includegraphics[width=\textwidth]{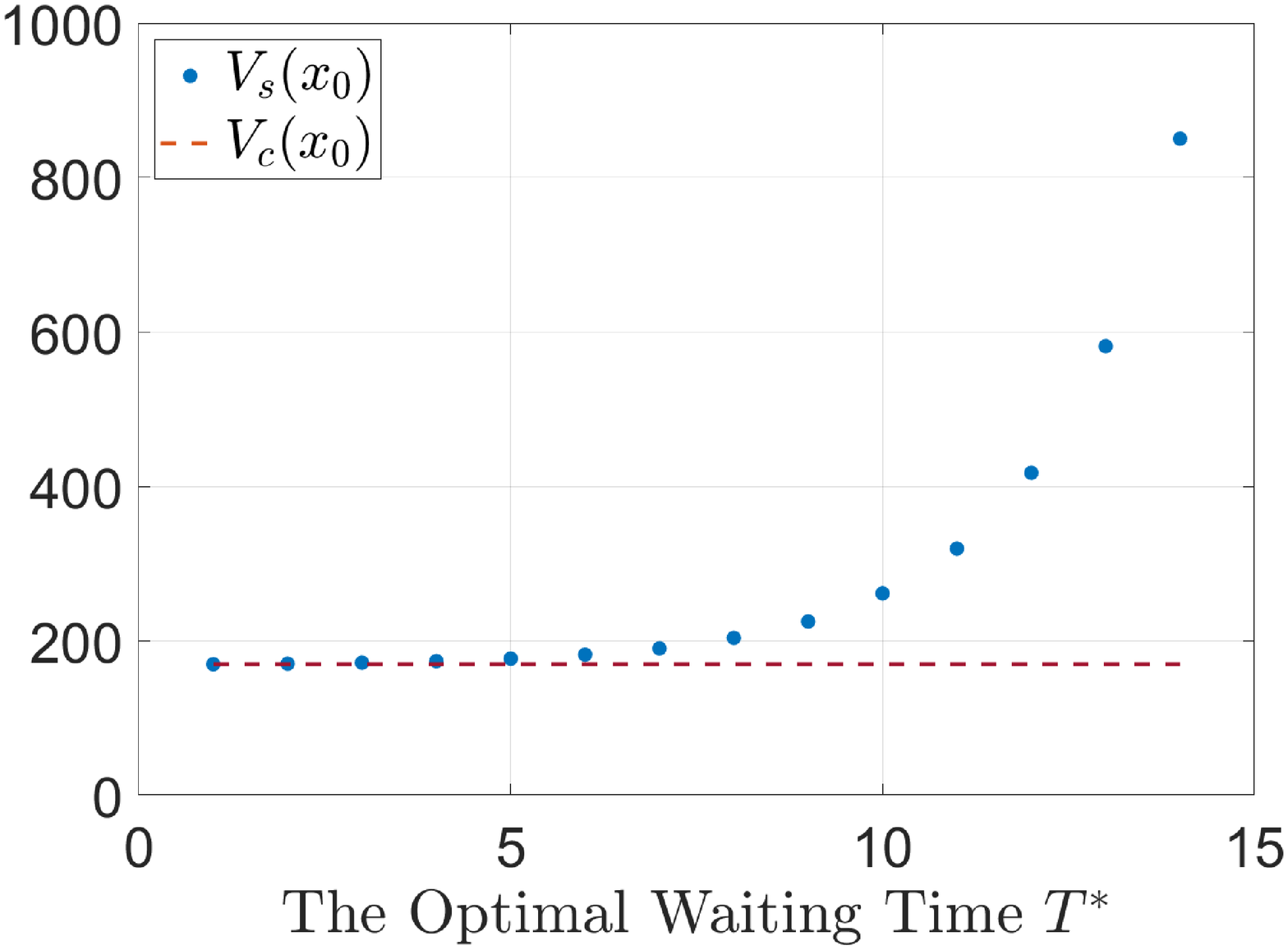}
         \caption{{\footnotesize An illustration of how controlled measurements affect the system performance (the cost excluding measurement costs).}}
         \label{fig:SysPerformVersusT}
     \end{subfigure}
        \caption{{\footnotesize Three Illustrations of the Overall Performance of the Optimal Measurement Strategy for \textbf{sys1}.}}
        \label{fig:three graphs}
\end{figure}

To compare different scenarios, we define the following quantities. Let $V_s(x_0)$ be the optimal system cost (cost excluding the cost of measurements) of the system starting at $x_0$. By definition and the results in\Cref{Eq:CharacterizedValueFunction}, 
$$
V_s(x_0)\coloneqq V(x_0) - \frac{\beta^{T^*}}{1-\beta^{T^*}} O = x'P x + \frac{\sum_{t=0}^{T^*-1} \beta^t \Tr\left( P_{t}(\bar{\mathcal{F}}_t)\varphi\right)}{1-\beta^{T^*}} + \frac{\beta}{1-\beta} \Tr\left( \Sigma_S C' PC\right),
$$
where $T^*$ is determined by $O$ according to \Cref{Eq:DetermineOptimalT}. Let $V_c(x_0)$ be the  optimal cost (value) of the classic LQG control problem, i.e., $V_c(x_0) \coloneqq x_0'Px_0 + {\beta}/{(1-\beta)} \Tr\left( \Sigma_S C' PC\right)$. Let $V_e(x_0)$ be the total cost when the measurement strategy is to measure every time. That is $V_e(x_0) \coloneqq x_0'Px_0 + O*\beta /(1-\beta)  $.

\begin{figure}[ht]
    \centering
    \includegraphics[width=1\columnwidth]{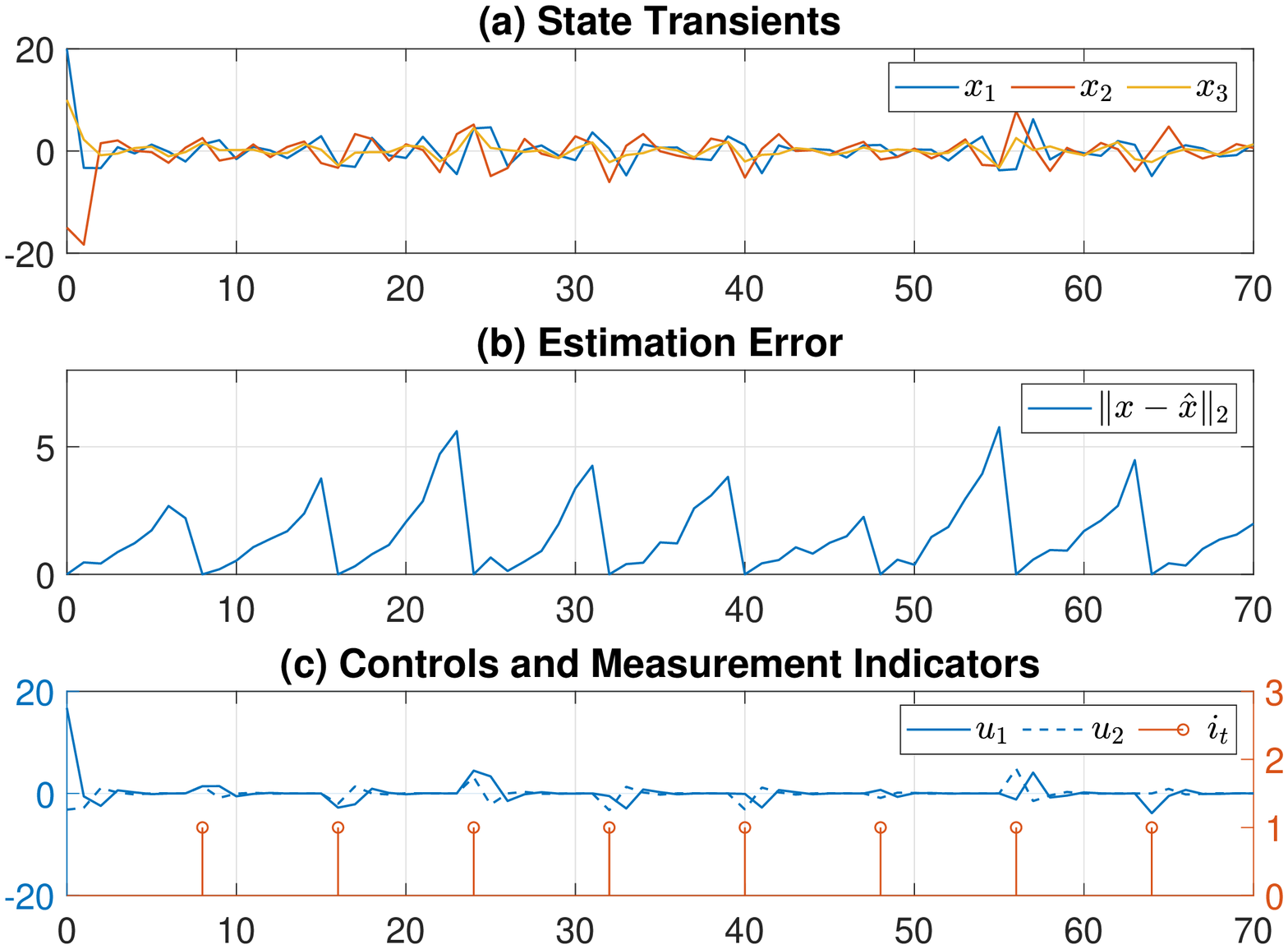}
    \caption{The dynamic behavior of \textbf{sys1} under the optimal measurement strategy when the cost of measurements is $50$.}
    \label{fig:DynamicBehaviorSys1O50}
\end{figure}

We have shown in \Cref{Theorem:CharacterizationRandOptimalObservation} that the optimal measurement strategy is to measure periodically and the optimal period length is determined by $O$. \Cref{fig:OVersusTstar} gives the relations between the cost of measurements $O$ and the optimal period length $T^*$ ($T^*$ is also called the optimal waiting time). It shows that even when the cost of measurement is relatively low (it is relatively low compared with the optimal cost of the classic problem $V_c(x_0) \coloneqq x_0'Px_0 + \beta\Tr\left(\Sigma_s C' P C\right)/(1-\beta) = 169.45$), the optimal measurement strategy suggests not measure every time. For example, when the cost of measurements is $10$, i.e., $O = 10$, the optimal measurement strategy is to measure every $6$ steps, $T^*=6$. That means the system performance is not degraded much even when the controller only chooses to measure once in $6$ steps. We can also see this point from \Cref{fig:SysPerformVersusT}, where the relations between the optimal cost excluding measurement costs $V_s(x_0)$ and the optimal waiting time $T^*$. We can see that when $T^*=6$ (corresponding to $O=10$), $V_s(x_0) = 176.65$. Compared with the strategy of measuring every time, the optimal measurement strategy only induces $(V_s(x_0)-V_c(x_0))/V_c(x_0) = 4.25\%$ degradation of the system performance. And more importantly, by following the optimal measurement strategy, i.e., measuring only once in $6$ steps, the controller can cut down $\beta O /(1-\beta) - \beta^{T^*} O /(1-\beta^{T^*}) = 0.95*10/0.05 - 0.95^6*10/(1-0.95^6)= 162.25$ cost of measurements. The cost of measurements saved constitutes $162.25/V(x_0)=79.38\%$ of the whole optimal cost $V(x_0)$. This shows the effectiveness of the optimal measurement strategy in reducing the overhead of measurements while keeping the system performance. To further compared the optimal measurement strategy with the strategy of measuring every time, we presents \Cref{fig:ComparedTwoStrategies}. The red dash line shows the total cost $V_e(x_0)$ of the problem when the controller chooses to measure every time. The blue line shows the optimal cost of the problem when the controller adopts the optimal measurement strategy. \Cref{fig:ComparedTwoStrategies} demonstrates that by adopting the optimal measurement strategy, the total cost will be reduced by a large quantity. And the larger the cost of measurements $O$, the more cost that the optimal measurement strategy can save.

\begin{figure}[ht]
    \centering
    \includegraphics[width=1\columnwidth]{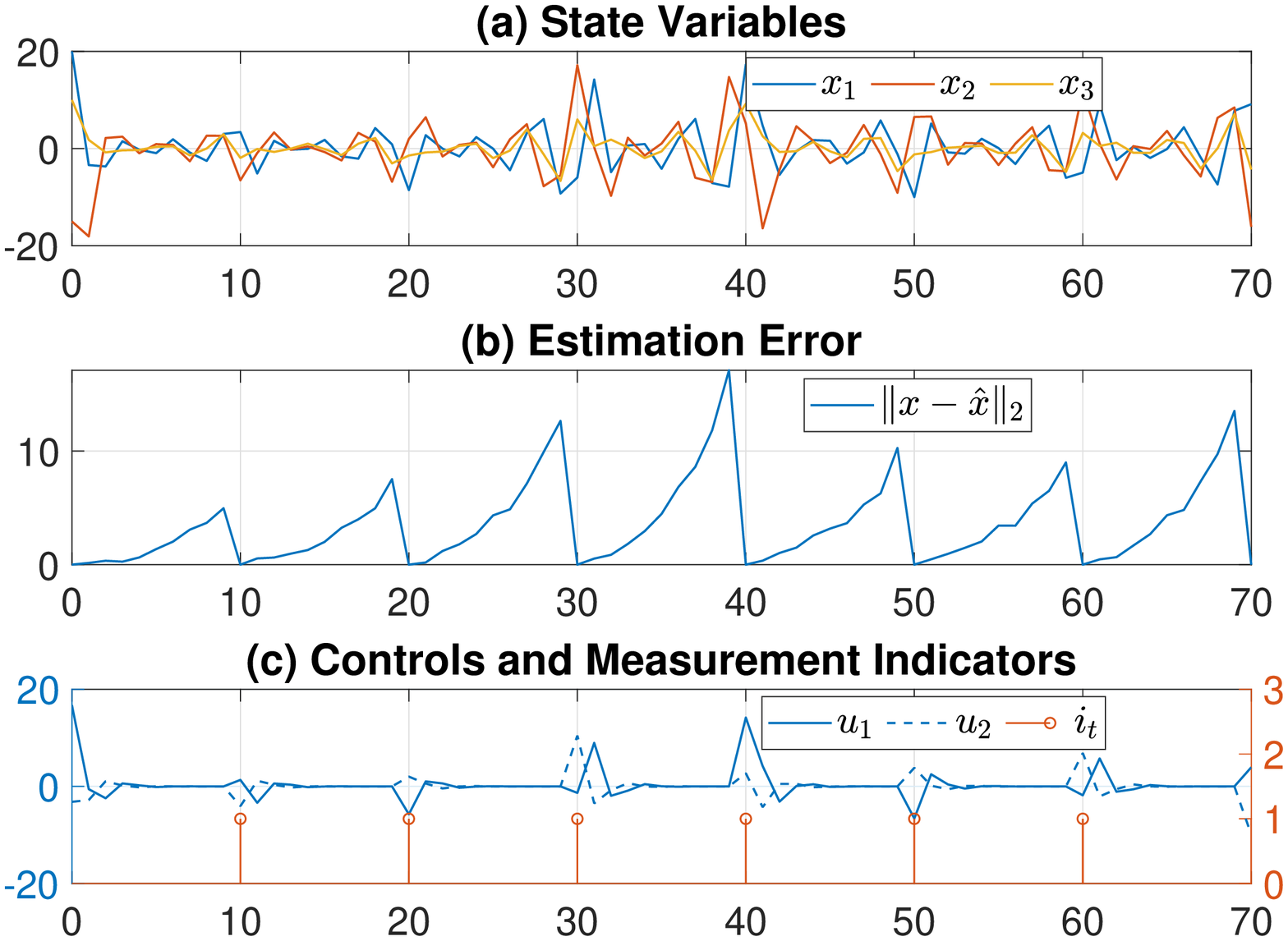}
    \caption{The dynamic behavior of \textbf{sys1} under the optimal measurement strategy when the cost of measurements is $300$.}
    \label{fig:DynamicBehaviorSys1O300}
\end{figure}

Note that the eigenvalues of $A_1$ have maximal magnitude $1.3561 >1$. Because the estimate error will be accumulated and amplified by $A_1$ if no measurement is made, the estimation quality deteriorate exponentially within a non-measurement interval, which will increases the system cost through the optimal control $u^*_t = -K\hat{x}_t$. Thus, from \Cref{fig:OVersusTstar}, we can see that the optimal waiting time grows linearly as the cost of measurements $O$ increases exponentially. Also, we can see, from \Cref{fig:SysPerformVersusT}, that the optimal system cost $V_s(x_0)$ increases exponentially as the optimal waiting time increases.

Next, we show the dynamic behavior of \textbf{sys1} under the optimal measurement strategy when the cost of measurements $O$ is $50$. When $O=50$, $T^*=8$. \Cref{fig:DynamicBehaviorSys1O50} presents the transitions of the state, the evolution of estimation error, and the selections of controls and measurements over $70$ steps. From \Cref{fig:DynamicBehaviorSys1O50}, we can see that the state is stabilized to the origin and evolves around the origin. The estimation error accumulates when there is no measurement and is cleared once a measurement is made. Between two measurements, the controls are open-loop controls with an initial condition equal to the last measured state. The open-loops controls are generated based on the estimate $\hat{x}$ which propagates like a noiseless system, i.e., $\hat{x}_{t+1}= A\hat{x}_t +Bu_t $ when there is no measurement. Then $\hat{x}_t$ tends to be zero if no measurement is made. Thus, as we can see from \Cref{fig:DynamicBehaviorSys1O50}, the controls tends to be zero until a new measurement is made. When the cost of measurements $O$ increases to $300$, $T^*=10$ and the dynamic behavior of \textbf{sys1} is shown in \Cref{fig:DynamicBehaviorSys1O300}. We can see that the state can still be stabilized to the origin but evolves around the origin with a larger margin. The estimation error accumulates to a higher magnitude before it is cleared by a measurement. The control still exhibits open-loop behavior (approaches zero when no measurement is made) between two measurements.

\begin{figure}[ht]
    \centering
    \includegraphics[width=1\columnwidth]{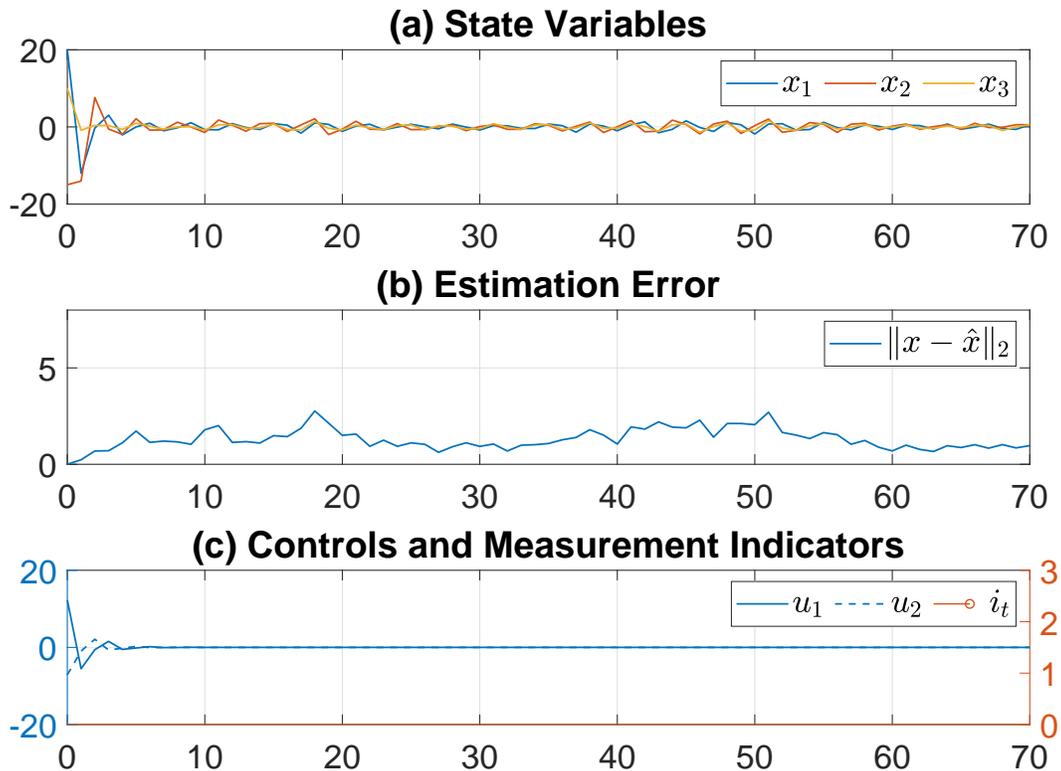}
    \caption{The dynamic behavior of \textbf{sys2} under the optimal measurement strategy when the cost of measurements is $7$.}
    \label{fig:DynamicBehaviorSys2O7}
\end{figure}

Lastly, we considers \textbf{sys2} where we have a Schur stable system matrix $A_2$. In this case, solving the Lyapunov function in \Cref{Eq:LypFun} for $W_\infty$ gives
$$
W_\infty= \begin{bmatrix}
    2.5129  & -0.8009  &  0.2130\\
   -0.8009  &  0.9080  &  0.5897\\
    0.2130  &  0.5897 &   0.8710\\
\end{bmatrix}.
$$

From 3) of \Cref{Theorem:CharacterizationRandOptimalObservation}, we know that if $ O\geq {\Tr\left( W_\infty \varphi\right)}/{(1-\beta)} -  \sum_{t=0}^\infty \beta^t \Tr\left( P_t(\bar{\mathcal{F}}_t)\varphi\right)$. For \textbf{sys2}, we have
$$
\frac{\Tr\left( W_\infty \varphi\right)}{(1-\beta)} -  \sum_{t=0}^\infty \beta^t \Tr\left( P_t(\bar{\mathcal{F}}_t)\varphi\right) \leq 
\frac{\Tr\left( W_\infty \varphi\right)}{(1-\beta)} - \sum_{t=0}^{1000}\beta^t \Tr\left( P_t(\bar{\mathcal{F}}_t) \varphi \right)= 6.4305.
$$
That means if the cost of measurements $O\geq 6.4305$, the optimal measurement strategy is to not measure at all. When the cost of measurements $O=7$, the optimal measurement strategy is to not measure at all. The dynamic behavior of \textbf{sys2} in this case is plotted in \Cref{fig:DynamicBehaviorSys2O7}. We can see that the no measurement is made; the controls are open-loop over the whole period and approach zero as time goes by. The estimation error accumulates but is diminished by a Schur stable $A_2$.

\section{Conclusions}
We addressed the co-design and co-optimization of an infinite horizon LQG control problem with costly measurements. We answered the questions of when is the optimal time to measure and how to control when having controlled measurements. The problem is central in modern control applications, such as IoT, IoBT, and control applications incorporated with SaaSs and CaaSs. The answers provide guidelines on designing a more economically efficient controller in such application scenarios and offer different alternatives for the controller to implement the optimal control and measurement strategies. We realized that the formulation of the representation problem defined by \Cref{Eq:CostFunctionalRep} has a natural application in the self-triggered control paradigm. The case when the controls are fixed between two measurements is discussed, and the results in \Cref{Theorem:DynamicProgramming} can be extended directly in this case. We leave the characterization of the optimal control and measurement strategies for future work.

The paper also opens several other avenues for future endeavours. First, the formulation can be studied and analyzed in a continuous-time LQG setting. A continuous-time setting allows us to choose the waiting time for the next measurement in a continuous space, i.e., $T\in (0,\infty]$ but also brings more issues when one needs to find the optimal waiting time. Second, the costly yet controlled measurement setting can be studied in a nonlinear system or a general MDP framework. In this case, the difficulty in deriving an analytical characterization of the optimal control and measurement strategies becomes prohibitive \cite{huang2019continuous}. Alternatively, we can resort to learning approaches by leveraging results in \Cref{Theorem:DynamicProgramming} and let the controller learn when to observe. An similar example is given in \cite{BieRaj2020}. Third, the controlled and costly measurements problem in LQG games has been studied in \cite{huang2020cross,maity2017linear}. However, only symmetric information problem has been investigated in \cite{huang2020cross,maity2017linear}, i.e., players co-decide whether to measure and receive the same measurement. An asymmetric information problem, where each player chooses to measure independently from other players and hence may receive measurements at different time steps than other players, may lead to more interesting discussions. 

\appendix

\subsection{Proof of \Cref{Lemma:EquivalentRepresentation}}\label{Proof:EquivalentRepresentation}
\begin{proof}
We prove the lemma by showing that every $\pi\in \Pi$ can be represented by a strategy $\tilde{\pi}\in\tilde{\Pi}$ and vice versa, and the represented strategy produces the same cost. 

At stage $t = 0$, since the initial state is disclosed to the controller, $i_0$ will be zero in any optimal solutions. Note that $\bar{T}_k$ denotes the time instance when the $k$th measurement being made, i.e., $I_{T_k}$ satisfies the following conditions: $i_{T_k} = 1$ and there are $k$ number of ones in $I_{T_k}$. For any $k$, let $t = \bar{T}_k$. Then $({T_{k+1}},u_t,\cdots, u_{t+T_{k+1}-1}) = \tilde{\pi}(x_t)$ is generated based on current observation $x_t$. This can be represented by the following policy 
$$
\begin{aligned}
i_{t+\tau} &= \mu(\mathcal{F}_{t+\tau}) = 0,\ \ \ \textrm{for }\tau = 1,2,\cdots,T_{k+1}-1,\\
i_{t+T_{k+1}} &= \mu({\mathcal{F}_{t+T_{k+1}}}) =1.
\end{aligned}
$$
Since the state-measurement $(x_n,y_n)$ dynamics defined in \Cref{Eq:SystemDynamics} is Markovian,  the latest state information in $\mathcal{F}_{t+\tau}$ for $\tau = 1,2,\cdots,T_{k+1}$ is $x_t = x_{\bar{T}_k}$. Hence, the controls $(u_t,\cdots, u_{t+T_{k+1}-1})$ is constructed based on $x_t$. That means the controls $(u_t,\cdots, u_{t+T_{k+1}-1})$ generated by $\tilde{\pi}(x_t)$ can also be represented by $(\nu(\bar{\mathcal{F}}_{t}), \nu(\bar{\mathcal{F}}_{t+1}),\cdots,\nu(\bar{\mathcal{F}}_{t+T_{t+1}-1}))$.

Conversely, let $I_t$ be the measurement indicators generated by a strategy $\pi\in \Pi$. Let $t$ be a time instance such that $i_t=1$ is the $k$th ones in $I_t$ and $t+T_{k+1}$ be a time instance such that $i_{t+T_{k+1}} =1$ is the $k+1$th ones in $I_{t+T_{k+1}}$. Note that the measurement being used to generate $u_t,\cdots,u_{t+T_{k+1}-1}$ $i_{t+1},\cdots,i_{t + T_{k+1}}$  is simply $y_t =x_t$. Thus, the strategy $\pi$ can be represented by $ \tilde{\pi}(x_t) = (T_{k+1}, \nu(\bar{\mathcal{F}}_t), \cdots,\nu(\bar{\mathcal{F}}_{t+T_{k+1}-1}))$. Hence, the two strategies are equivalent representations of each other. It is easy to see that the strategy $\pi$ produces the same cost under \Cref{Eq:CostFunctional} as the represented strategy $\tilde{\pi}$ under \Cref{Eq:CostFunctionalRep}, and vice versa. In fact, given any sequence of measurement indicators with $i_0=0$ (it is assumed that the initial condition is known to the controller), we can write the last term of \Cref{Eq:SystemDynamics} as
$$
\sum_{t=0}^\infty \beta^t i_t O = \sum_{t=0}^{\infty} \beta^t \mathds{1}_{\{i_t=1\}}O = \sum_{k=1}^\infty \beta^{\bar{T}_k}O.
$$
This produces the last term of \Cref{Eq:CostFunctionalRep}.

\end{proof}

\subsection{Proof of \Cref{Theorem:DynamicProgramming}} \label{Proof:DynamicProgramming}
\begin{proof}
We prove the theorem by constructing a consolidated Markov decision process problem where the costs induced, the controls generated between observation epoch are considered as a stage cost and a concatenated control. Let $\bar{c}_k$ be the sum of the costs induced between the $k$th measurement and $k+1$th measurement by policy $\tilde{\pi}$. That is
$$
\bar{c}_k = \bar{c}\left(x_{\bar{T}_k},\tilde{\pi}(x_{\bar{T}_k})\right) =  \bar{c}(x(\bar{T}_k),T_{k+1},u_{\bar{T}_k: \bar{T}_k+ T_{k+1}-1}) = \mathbb{E}\left[ \sum_{t=\bar{T}_k}^{\bar{T}_k + T_{k+1} -1} \beta^{t-\bar{T}_k} (x_t' Q x_t + u_t' R u_t) \middle \vert x(\bar{T}_k), \pi\left(x(\bar{T}_k) \right) \right].
$$
By Fubini's Theorem and Markov property \cite{durrett2019probability}, we have
$$
\bar{c}\left(x,\tilde{\pi}(x) \right) = \bar{c}\left(x,T,u_{0:T-1} \right)= \sum_{t=0}^{T-1} \beta^t \mathbb{E}\left[ x_t' Q x_t + u_t' R u_t \middle \vert x_0 =x, \tilde{\pi}(x)\right]. 
$$
Then, $\tilde{F}(\tilde{\pi};x)$ can be reformulated as 
\begin{equation}\label{Eq:ReformedCostFunctional1}
\tilde{F}(\tilde{\pi};x) = \mathbb{E}\left[ \sum_{k=0}^\infty \beta^{\bar{T}_k} (\bar{c}_k + \beta^{T_k} O) \middle \vert x_0= x, \tilde{\pi}\right].
\end{equation}
A close look at \Cref{Eq:ReformedCostFunctional1} shows that this is a discounted cost discrete-time Markov decision process with discounted factor $\beta$, Markov state and Markovian actions given respectively by
$$
Z_k\coloneqq (x_{\bar{T}_k}, \tilde{T}_k), A_k=(T_{k+1}, u_{\bar{T}_k:\bar{T}_k +T_{k+1}-1}),
$$
where $\tilde{T}_k \coloneqq \bar{T}_k - k$, and running cost equal to
$$
C(Z_k,A_k) = \beta^{\tilde{T}_k}\left[ \bar{c}\left( x_{\bar{T}_k}, T_{k+1}, u_{\bar{T}_k:\bar{T}_k+ T_{k+1}-1} + \beta^{T_k} O \right)\right].
$$
That is, cost in \Cref{Eq:ReformedCostFunctional1} is given by
$$
\tilde{F}(\tilde{\pi};x) = \mathbb{E}\left[ \sum_{k=0}^\infty \beta^{k} C(Z_k,A_k) \middle \vert Z_0 = (x,0) \right].
$$
The consolidated formulation can be treated as a regular Markov decision problem and hence the results (mainly the results available to Polish spaces) can be derived from current Markov decision literature. By Theorem 6.2.7, the claims in \Cref{Theorem:DynamicProgramming} follow immediately.
\end{proof}

\subsection{Proof of \Cref{Lemma:InnerOptimalControl}}\label{Proof:InnerOptimalControl}
\begin{proof}
Given that $V(x) = x'P x + r$ and $T$ is fixed, the inner minimization problem in \Cref{Eq:DynamicProgramming} can be considered as an open-loop optimal control problem with cost functional 
\begin{equation}\label{Eq:InnerCostFunctional}
\inf_{u_0,\cdots,u_{T-1}} \mathbb{E} \left[ \sum_{t=0}^{T-1} \beta^t (x_t' Q x_t + u_t' R u_t) + \beta^T x_T'Px_T + \beta^T r + \beta^T O \middle \vert x_0 =x \right],
\end{equation}
and system dynamics \Cref{Eq:SystemDynamics}. Let $\bar{\mathcal{F}}_k$ be the information available at time $k$ defined in \Cref{Eq:InformationStructure} corresponding to the measurement sequence $i_1 =0, i_2 =0,\cdots,i_{T-1}=0, i_T =1$. Define the cost-to-go functional of the optimal control problem in \Cref{Eq:InnerCostFunctional} as
$$
f_k(x) = \mathbb{E} \left[ \sum_{t=k}^{T-1} \beta^t (x_t' Q x_t + u_t' R u_t) + \beta^T x_T'Px_T + \beta^T r + \beta^T O \middle \vert \bar{\mathcal{F}}_k \right].
$$
Define the optimal cost-to-go functional as $f_k^*(x) = \inf_{u_{k:T-1}} f_k(x)$. An application of dynamic programming techniques yields
$$
f_k^*(x) = \min_{u_k} \mathbb{E}\left[ \beta^k (x_k'Qx_k + u_k' R u_k') + f_{k+1}^*(x)\middle\vert \bar{\mathcal{F}}_k \right].
$$
By definition, $f_{T}(x)^* = f_{T}(x) = \mathbb{E}\left[  \beta^T(x_T' P x_T + r +  O) \middle \vert \bar{\mathcal{F}}_{T} \right]$. At $k = T-1$, we have
\begin{equation}\label{Eq:InnerDPEquation}
f_{T-1}^*(x) = \min_{u_{T-1}} \beta^{T-1}  \mathbb{E}\left[ (x_{T-1}' Q x_{T-1} +  u_{T-1}' R u_{T-1}) + \beta (x_T' P x_T) + \beta(r+O)\middle \vert \bar{\mathcal{F}}_{T-1} \right].
\end{equation}
Substituting $x_T = A x_{T-1} + B u_{T-1} + C w_{T-1}$ into $f_{T-1}^*$ and solving the minimization problem for $u^*_{T-1}$ yields
$$
u_{T-1}^* = -(R+ \beta B'L_0 B)^{-1} \beta B' L_0 A \hat{x}_{T-1},
$$
and applying $u^*_{T-1}$ in $f_{T-1}^*$ gives
$$
\begin{aligned}
f_{T-1}^* = &\beta^{T-1} \Big\{ \mathbb{E}\left[ x_{T-1}' (Q + \beta A'L_0 A - A' L_0 B \beta (R+\beta B'L_0 B)^{-1}\beta B'L_0 A)x_{T-1} \middle \vert \bar{\mathcal{F}}_{T-1}  \right]\\
& +\mathbb{E}\left[ (x_{T-1} - \hat{x}_{T-1})' A' L_0 B\beta (R+\beta B'L_0 B)^{-1} \beta B'L_0A (x_{T-1} -\hat{x}_{T-1})' \middle \vert \bar{\mathcal{F}}_{T-1}\right]\\
&+ \beta \mathbb{E}\left[ w_{T-1}' C'L_0 C w_{T-1}  \middle \vert \bar{\mathcal{F}}_{T-1}\right] + \beta (r+O)
\Big\}\\
=&  \beta^{T-1} \Big\{ \mathbb{E} \left[ x_{T-1}' L_1 x_{T-1} \middle \vert \bar{\mathcal{F}}_{T-1} \right] + \Tr\left( P_{T-1}(\bar{\mathcal{F}}_{T-1}) \varphi_{T-1}\right)  + \Tr\left( \Sigma_S C'L_0C \right)+ \beta (r+O) \Big\},
\end{aligned}
$$
where $L_1$ agrees with \Cref{Eq:RiccatiEquation} and $\varphi_{T-1}$ agrees with \Cref{Eq:EstimationErrorCoefficient}. The cases for $k= T-2$ till $k=0$ can be conducted similarly through induction using the inner dynamic programming equation \Cref{Eq:InnerDPEquation}. 
\end{proof}

\subsection{Proof of \Cref{Lemma:ValueFunctionCharacterization}}\label{Proof:ValueFunctionCharacterization}
\begin{proof}
From Theorem 4 in Section 9.3.2 of \cite{kushner_introduction_1971}, we know that if $(A,B)$ is controllable, $L_0,L_1,\cdots, L_T$ generated by the Riccati equation \Cref{Eq:RiccatiEquation} is non-decreasing, i.e., $L_0 \leq L_1\leq\cdots \leq L_T$. Note that $L_0 = P$. For any $T\in\mathbb{N}$, $L_T = P$ implies $L_0 = L_1 = \cdots = L_T=P$. That means the dynamic programming equation \Cref{Eq:DPInnerSolved} holds if and only if $P$ satisfies the algebraic Riccati equation \Cref{Eq:AlgebraicRiccatiEquation}. According to Theorem 4 in Section 9.3.2 of \cite{kushner_introduction_1971}, the algebraic Riccati equation admits a unique positive definite solution if  $(A,L)$ is observable. Since now we have $L_0= L_1 = \cdots =L_T = P$, $\varphi_t = \varphi$ in \Cref{Eq:DPInnerSolved} for $t = 0,\cdots,T-1$, where
$
\varphi= A' PB \beta (R+\beta B' P B)^{-1}\beta B' PA.
$
With $P$ be characterized, we can write \Cref{Eq:DPInnerSolved} as
\begin{equation}\label{Eq:DPQuadraticSolved}
x'Px + r  =  x' P x  + \inf_{T\in\mathbb{N}} \left\{\sum_{t=0}^{T-1} \beta^t \Tr\left( P_{t}(\bar{\mathcal{F}}_t)  \varphi\right) + \sum_{t=1}^{T} \beta^t \Tr\left( \Sigma_S C' PC\right) + \beta^T (r+O) \right\}.
\end{equation}
It is easy to see that $r$ is the solution of the fixed-point equation defined in \Cref{Eq:FixedPointEquationForr}, whose existence and uniquess are guaranteed by Banach fixed-point theorem \cite{kreyszig1978introductory}.
\end{proof}

\subsection{Proof of \Cref{Theorem:CharacterizationRandOptimalObservation}}\label{Proof:CharacterizationRandOptimalObservation}
\begin{proof}
Define a function of $T$ as 
$$
f(T) = \sum_{t=0}^{T-1} \beta^t\Tr\left( P_t(\mathcal{F}_t)\varphi\right) + \sum_{t=1}^T  \beta^t \Tr\left(\Sigma_S C'PC \right)+ \beta^T(r+O).
$$
Note that $f(T)$ is also depends on $r$. Here, we write $f(T)$ for national simplicity. The fixed-point equation \Cref{Eq:FixedPointEquationForr} can then be written as $r = \inf_{T\in \mathbb{N}} f(T)$. To find $T^*$ that minimizes $f(T)$, we calculate
\begin{equation}\label{Eq:DifferenceofF}
\begin{aligned}
f(T+1) - f(T) &= \beta^T \Tr\left( P_T(\bar{\mathcal{F}}_T) \varphi\right) + \beta^{T+1}\Tr\left( \Sigma_S C'PC\right) + (\beta^{T+1} - \beta^T)(r+O)\\
&=\beta^T \left[ \Tr\left( P_{T}(\bar{\mathcal{F}}_T)\varphi \right) + \beta \Tr\left( \Sigma_S C'PC \right)- (1- \beta) (r+O)\right]\\
&= \beta^T \left[ \Tr\left( \sum_{t=0}^{T-1} (A')^t C\Sigma_S C A^{t} \varphi  \right)+ \beta \Tr\left( \Sigma_S C'PC\right) - (1- \beta) (r+O)\right],
\end{aligned}
\end{equation}
where the last equality is obtained using the fact that $P_{T}(\bar{\mathcal{F}}_T) =\sum_{t=0}^{T-1} (A')^t C\Sigma_S C A^{T-1}$. Note that the term in the square brackets in \Cref{Eq:DifferenceofF} 
$$
h(T) = \Tr\left( \sum_{t=0}^{T-1} (A')^t C'\Sigma_S C A^{t} \varphi \right) + \beta \Tr\left( \Sigma_S C'PC\right) - (1- \beta)(r+O)
$$
is strictly increasing in $T$. Thus, if $h(1)> 0$, then $h(T) > 0$ for all $T >1$. If $h(\infty)$ exists and $h(\infty)\leq 0$, $h(T)<0$ for all $T<\infty$. Otherwise, there exists a $T^*$ such that $h(T^*-1)<=0$ and $h(T^* )> 0$. Since $h(T)$ is strictly increasing in $T$, we have $h(T)<0$ for all $T <T^*-1$ and $h(T)>0$ for all $T>T^*$. Since $f(T+1) - f(T) = \beta^T h(T)$, we can see that if $h(1) > 0$, the optimal waiting time for next observation is $T^*=1$; If $h(\infty)<=0$, the optimal policy is to not measure at all; If there exists a $T^*$ such that $h(T^*-1)<=0$ and $h(T^*)>0$, the optimal measurement policy is $T^*$.

First, we discuss the case when $h(1)>=0$. We have $f(T+1) - f(T) >0$ for all $T$. Thus, $T^* =1$, which means the optimal measurement policy is to measure every time. By \Cref{Eq:FixedPointEquationForr}, we have
$$
r  =   \beta \Tr\left( \Sigma_S C' PC \right)+ \beta (r+O),
$$
which gives $r = \frac{\beta}{1-\beta} \Tr\left( \Sigma_S C' PC\right) + \frac{\beta}{1- \beta} O$. Also note that $h(1)>0$ implies that
$$
\Tr\left(  C\Sigma_S C \varphi\right)  + \beta \Tr\left( \Sigma_S C'PC \right)- (1- \beta)(r+O) > 0.
$$
Using the value of $r$, we have 
$$
\begin{aligned}
(1-\beta)( \frac{\beta}{1-\beta} \Tr\left( \Sigma_S C' PC\right) + \frac{\beta}{1- \beta} O +O) &< \Tr\left( C\Sigma_S C \varphi\right) + \beta \Tr\left(\Sigma_S C'PC\right)\\
O&< \Tr\left( C\Sigma_S C\varphi\right).
\end{aligned}
$$
Thus, we can say that when $O<\Tr\left( C'\Sigma_S C\varphi\right)$, the value function is $V(x) = x'Px + r$ where $P$ is the solution of \Cref{Eq:AlgebraicRiccatiEquation} and $r = \frac{\beta}{1-\beta} \Tr\left( \Sigma_S C'PC\right) + \frac{\beta}{1-\beta} O$; the optimal measurement policy is to observe every time, $T^* =1 $.

Second, we discuss the case when there exists a $T^*$ such that $h(T^*-1)<=0$ and $h(T^*) > 0$. In this case, the optimal measurement policy is $T^*$. By equation \Cref{Eq:FixedPointEquationForr}, we have
$$
r = \sum_{t=0}^{T^*-1} \beta^t\Tr\left( P_t(\bar{\mathcal{F}}_t)\varphi\right) + \sum_{t=1}^{T^*}  \beta^t \Tr\left(\Sigma_S C'PC \right)+ \beta^{T^*}(r+O),
$$
which yields
\begin{equation}\label{Eq:ValueOfRWhenTstar}
r = \frac{\sum_{t=0}^{T^*-1} \beta^t \Tr\left( P_{t}(\bar{\mathcal{F}}_t)\varphi\right)}{1-\beta^{T^*}} + \frac{\beta}{1-\beta} \Tr\left( \Sigma_S C' PC \right)+ \frac{\beta^{T^*}}{1-\beta^{T^*}} O.
\end{equation}
Besides, $h(T^*-1)< 0$ and $h(T^*)\geq 0$ yields
{\small$$
\begin{aligned}
 \frac{\Tr\left( P_{T^*-1}(\bar{\mathcal{F}}_{T^*-1})\varphi\right)}{1-\beta} + \frac{\beta}{1-\beta} \Tr\left( \Sigma_S C'PC\right) -r &\leq  O < \frac{\Tr\left( P_{T^*}(\bar{\mathcal{F}}_{T^*})\varphi\right)}{1-\beta} + \frac{\beta}{ 1- \beta} \Tr\left( \Sigma_S C'PC \right)- r\\
 \frac{1-\beta^{T^*}}{1-\beta} \Tr\left( P_{T^*-1}(\bar{\mathcal{F}}_{T^*-1}) \varphi\right) - \sum_{t=0}^{T^* -1} \beta^t \Tr\left( P_t(\bar{\mathcal{F}}_t) \varphi\right) &\leq O < \frac{1-\beta^{T^*}}{1-\beta} \Tr\left( P_{T^*}(\bar{\mathcal{F}}_{T^*}) \varphi\right) - \sum_{t=0}^{T^* -1} \beta^t \Tr\left( P_t(\bar{\mathcal{F}}_t) \varphi\right)\\
   \end{aligned}
$$
$$
\begin{aligned}
 \sum_{t=0}^{T^*-1} \beta^t  \Tr\left(\left[ P_{T^*-1}(\bar{\mathcal{F}}_{T^*-1}) - P_t(\bar{\mathcal{F}}_t)  \right]\varphi\right) &\leq O < \sum_{t=0}^{T^*-1} \beta^t  \Tr\left(\left[ P_{T^*}(\bar{\mathcal{F}}_{T^*}) - P_t(\bar{\mathcal{F}}_t)  \right]\varphi\right)\\
 \sum_{t=0}^{T^*-2} \beta^t  \Tr\left(\left[ P_{T^*-1}(\bar{\mathcal{F}}_{T^*-1}) - P_t(\bar{\mathcal{F}}_t)  \right]\varphi\right) &\leq O < \sum_{t=0}^{T^*-1}\beta^t \Tr\left(\left[ P_{T^*}(\bar{\mathcal{F}}_{T^*}) - P_t(\bar{\mathcal{F}}_t)  \right]\varphi\right)\\
 \sum_{t=0}^{T^*-2} \beta^t \Tr\left( \left[ \sum_{t=0}^{T^*-2} (A')^{t} C' \Sigma_s C A^t -\sum_{\tau=0}^{t-1} (A')^\tau C' \Sigma_s C A^\tau \right]\varphi\right) & \leq O  <  \sum_{t=0}^{T^*-1} \beta^t \Tr\left( \left[ \sum_{t=0}^{T^*-1} (A')^{t} C' \Sigma_s C A^t -\sum_{\tau=0}^{t-1} (A')^\tau C' \Sigma_s C A^\tau \right]\varphi\right) \\
  \end{aligned}
$$
$$
\begin{aligned}
 \sum_{t=0}^{T^*-2} \beta^t \Tr\left( \left[ \sum_{\tau =t}^{T^*-2} (A')^\tau C'\Sigma_S C A^\tau \right] \varphi\right) &\leq O < \sum_{t=0}^{T^*-1} \beta^t \Tr\left( \left[ \sum_{\tau=t}^{T^*-1} (A')^\tau C' \Sigma_S C A^\tau \right]\varphi\right)\\
 \sum_{t=0}^{T^*-2} \frac{1-\beta^{t+1}}{1-\beta} \Tr\left( (A')^tC' \Sigma_S CA^t \varphi\right) &\leq O < \sum_{t=0}^{T^*-1} \frac{1-\beta^{t+1}}{1-\beta} \Tr\left( (A')^tC' \Sigma_S CA^t \varphi\right).
\end{aligned}
$$}
Hence, we can conclude that given the cost of measurement $O$, there optimal measurement waiting time is $T^*$ that satisfies $ \sum_{t=0}^{T^*-2} \frac{1-\beta^{t+1}}{1-\beta} \Tr\left( (A')^tC' \Sigma_S CA^t \varphi\right) \leq O < \sum_{t=0}^{T^*-1} \frac{1-\beta^{t+1}}{1-\beta} \Tr\left( (A')^tC' \Sigma_S CA^t \varphi\right)$. The value function is $V(x) = x'Px + r$ where $P$ is the solution of \Cref{Eq:AlgebraicRiccatiEquation} and $r$ is given by \Cref{Eq:ValueOfRWhenTstar}.

Now it remains to discuss $h(T)$ as $T$ goes to infinity. We first introduce the claim that shows the boundedness of $\sum_{t=0}^{T-1} \Tr\left( (A')^t C'\Sigma_S C A^t \varphi\right)$.

\begin{claim}
Suppose $\varphi$ is positive definite. The sum $\sum_{t=0}^{T-1} \Tr\left( (A')^t C'\Sigma_S C A^t \varphi\right)$ will converge if and only if all eigenvalues of $A$ have magnitude strictly smaller than $1$. 
\end{claim}
\begin{proof}
Define a matrix norm $\Vert \cdot \Vert$ as 
$$
\Vert M \Vert = \sqrt{\Tr\left( M' \varphi^{1/2}C'\Sigma_S C\varphi^{1/2}M\right)}.
$$
The norm is well defined since $C'\Sigma_S C$ and $\varphi$ are positive definite. Note that 
$$
\begin{aligned}
\Tr\left(  (A')^tC'\Sigma_SCA^t \varphi\right)&= \Tr\left(\varphi^{1/2}  (A')^t C'\Sigma_S CA^t \varphi^{1/2}\right)\\
&=\Tr\left( [(\varphi^{-1/2} A \varphi^{1/2})']^t \varphi^{1/2} C'\Sigma_S C\varphi^{1/2} [\varphi^{-1/2} A \varphi^{1/2}]^t \right)\\
&=\left\Vert \left(\varphi^{-1/2} A \varphi^{1/2}\right)^t \right\Vert^2.\\
\end{aligned}
$$
Note that $\varphi^{-1/2} A \varphi^{1/2}$ has the same eigenvalues as $A$. With Gelfand's formula \cite{lax_linear_2007}, one has
$$
\rho(M) = \lim_{k\rightarrow \infty} \Vert M^k \Vert^{1/k},
$$
where $\rho(M)$ is the spectral radius of matrix $M$. Using Gelfand's formula, one can shows by the root test that the sum $$\sum_{t=0}^{T-1} \Tr\left((A')^t C'\Sigma_S C A^t \varphi\right) = \sum_{t=0}^{T-1} \left\Vert \left(\varphi^{-1/2} A \varphi^{1/2}\right)^t \right\Vert^2$$
will diverge when $A$ has any eigenvalue of magnitude strictly greater than $1$ and will converge when all eigenvalues of $A$ have magnitude strictly less than $1$ ($A$ is stable). 

When $A$ has an eigenvalue of maximal magnitude $1$, then the sum also diverges. To see this, if $v$ is a unit eigenvector of $\varphi^{-1/2} A \varphi^{1/2}$ associated with eigenvalue $\lambda$ with $|\lambda|=1$, then we have
$$
\begin{aligned}
\left\Vert \left(\varphi^{-1/2} A \varphi^{1/2}\right)^t \right\Vert^2 &\geq v' [(\varphi^{-1/2} A \varphi^{1/2})']^t \varphi^{1/2} C'\Sigma_S C\varphi^{1/2} [\varphi^{-1/2} A \varphi^{1/2}]^t v\\
&= |\lambda|^{2t} v' \varphi^{1/2} C'\Sigma_S C\varphi^{1/2} v >0,
\end{aligned}
$$
which indicates that the sequence being added has a positive lower bound. Hence, the sum necessarily diverges. This completes our proof.
\end{proof}
Note that even if $\varphi$ is not positive definite, $\sum_{t=0}^{T-1} \Tr\left( (A')^t C'\Sigma_S C A^t \varphi\right)$ has a limit when $A$ has only eigenvalues with magnitude strictly less than $1$. From \cite{kushner_introduction_1971}, we know that for $A$ stable, the Observability Gramian
$$
W_\infty = \sum_{t=0}^\infty (A')^t C'\Sigma_S C A^t
$$
is the unique solution of the Lyapunov equation
$$
W_\infty - A' W_\infty A = C'\Sigma_S C.
$$
Hence, $\sum_{t=0}^{T-1} \Tr\left( (A')^t C'\Sigma_S C A^t \varphi\right) \rightarrow  \Tr\left( W_\infty \varphi\right)$. 

From the discussion above, we can conclude that when $A$ is unstable and $\varphi$ is positive definite, the optimal waiting time for next measurement $T^*$ is bounded $T^* <\infty$. That means when $A$ is unstable, the controller has to measure once in a finite period of time. When $A$ is table,
$$
h(\infty)= \Tr\left( W_\infty \varphi\right) + \beta \Tr\left( \Sigma_S C'PC\right) - (1-\beta)(r+O).
$$
We know that if $h(\infty) <= 0$, the best measurement policy is to not measure at all, i.e., $T^* =\infty$. In this case, we have
$$
0<\frac{\Tr\left( W_\infty \varphi\right)}{1-\beta} -  \sum_{t=0}^\infty \beta^t \Tr\left( P_t(\bar{\mathcal{F}}_t) \varphi\right) \leq O.
$$

Thus, we can conclude that if $A$ is stable and $O\geq \frac{\Tr\left( W_\infty \varphi\right)}{1-\beta} -  \sum_{t=0}^\infty \beta^t \Tr\left( P_t(\bar{\mathcal{F}}_t)\varphi\right)$, the best strategy is to not measure at all, i.e., $T^*=\infty$. The value function then is
$$
V(x) = x'Px +r,
$$
where $P$ is the solution of \Cref{Eq:AlgebraicRiccatiEquation} and $r= \sum_{t=0}^\infty \beta^t \Tr\left( P_t(\bar{\mathcal{F}}_t) \varphi\right) + \frac{\beta}{1-\beta} \Tr\left( \Sigma_S C'PC\right)$.
\end{proof}

% or


\ifCLASSOPTIONcaptionsoff
  \newpage
\fi



\bibliographystyle{IEEEtran}
\bibliography{ControlledMeasurement.bib}
\end{document}